\documentclass{llncs}
\usepackage[T1]{fontenc}
\usepackage[utf8]{inputenc}
\usepackage{graphicx}
\usepackage{amsmath}
\usepackage{xcolor}
\usepackage{xspace}
\usepackage{listings}
\usepackage[hidelinks]{hyperref}
\hypersetup{breaklinks=true}
\usepackage[detect-all]{siunitx}
\usepackage[capitalise]{cleveref}
\usepackage{float}
\usepackage{algorithm}
\usepackage[noend]{algpseudocode}
\usepackage{bytefield}
\usepackage{tikz}
\usepackage{placeins}
\usepackage{pgfplots}\pgfplotsset{compat=1.18}
\usepackage{fontawesome}
\usepackage{pifont}
\usepackage{fp}
\usepackage{array}
\usepackage{tabularx}
\usepackage{booktabs}
\usepackage{subcaption}

\usepackage{multirow}
\usepackage{colortbl}

\usepackage{enumitem}

\usepackage[normalem]{ulem}
\usepackage{comment}

\usetikzlibrary{shapes,arrows,positioning,calc,fit,patterns}
\usepgfplotslibrary{groupplots,statistics}

\newcommand{\sysname}[0]{\texorpdfstring{\textsc{Phy\-sa\-lia}}{Physalia}\xspace}
\newcommand{\sharingname}[0]{\texorpdfstring{\sysname-RCSS}{Physalia-RCSS}\xspace}

\newcommand{\outsym}[1]{\ensuremath{y_{#1}}}
\newcommand{\share}[1]{\ensuremath{s_{#1}}}

\newcommand{\sig}[1]{\ensuremath{\sigma_{#1}}}
\newcommand{\sign}[2]{\ensuremath{\mathrm{sign}_{#1}(#2)}}
\newcommand{\verify}[3]{\ensuremath{\mathrm{verify}_{#1}(#2, #3)}}

\newcommand{\RSSShare}[0]{\ensuremath{\mathsf{RSS.Share}}}
\newcommand{\RSSReconstruct}[0]{\ensuremath{\mathsf{RSS.Reconstruct}}}

\newcommand{\css}[0]{\ensuremath{\Pi}}
\newcommand{\CSSShare}[0]{\ensuremath{\mathsf{CSS.Share}}}
\newcommand{\CSSRecover}[0]{\ensuremath{\mathsf{CSS.Recover}}}

\newcommand{\redbw}[0]{\ensuremath{\tau}}

\ifdefined\SuppressAnnotations
  \newcommand{\annot}[3]{}
\else
  \newcommand{\annot}[3]{{\color{#1}{$\rule{8pt}{8pt}_{\textsf{\bfseries #2}}$ #3}}}
\fi

\ifdefined\HideChanges
  \newcommand{\del}[1]{}
  
  \newcommand{\moved}[1]{}
  \newcommand{\delfootnote}[1]{}
  \excludecomment{delblock}
\else
  \newcommand{\del}[1]{{\color{red}\sout{#1}}}
  
  \newcommand{\moved}[1]{{\color{red}\itshape[#1]}}
  \newcommand{\delfootnote}[1]{\footnote{{\color{red}\sout{#1}}}}
  
\fi

\makeatletter
\renewcommand\paragraph{\@startsection{paragraph}{4}{\z@}%
  {-6\p@ \@plus -2\p@ \@minus -2\p@}%
  {-0.5em \@plus -0.22em \@minus -0.1em}%
  {\normalfont\normalsize\itshape}}
\makeatother

\definecolor{bluefte}{HTML}{008fd5}
\definecolor{redfte}{HTML}{fc4f30}
\definecolor{yellowfte}{HTML}{e5ae38}
\definecolor{greenfte}{HTML}{6d904f}
\definecolor{greyfte}{HTML}{8b8b8b}
\definecolor{purplefte}{HTML}{810f7c}

\usepackage{graphicx}
\usepackage{xcolor}

\definecolor{greyfte}{HTML}{8b8b8b}

\definecolor{redfte}{HTML}{e41a1c}
\definecolor{bluefte}{HTML}{377eb8}
\definecolor{greenfte}{HTML}{4daf4a}
\definecolor{purplefte}{HTML}{984ea3}
\definecolor{yellowfte}{HTML}{ff7f00}

\colorlet{ctotal}{bluefte}
\colorlet{crequest}{redfte}
\colorlet{cresponse}{yellowfte}

\pgfplotsset{compat = 1.18}

\pgfplotsset{
  legend style = {font=\normalfont},
  tick label style = {font=\normalfont}
}

\pgfplotsset{dotted/.style={dash pattern=on 0pt off 2\pgflinewidth}}
\pgfplotsset{dashed/.style={dash pattern=on 2\pgflinewidth off 2\pgflinewidth}}
\pgfplotsset{dashdotdotted/.style={
    dash pattern=on 0pt off 2\pgflinewidth on 0pt off 2\pgflinewidth on 2\pgflinewidth
    off 2\pgflinewidth
  }
}
\pgfplotsset{dashdotted/.style={
    dash pattern=on 0pt off 2\pgflinewidth on 2\pgflinewidth off 2\pgflinewidth
  }
}

\tikzset{
}

\pgfplotsset{
  grid = none,
  every axis plot/.style={
    thick,
    line cap = round,
    mark options = {
      style = solid,
      fill = white,
      scale=0.8,
    },
  },
  every axis/.style={
    clip=false,
    cycle list name = bestnodash,
    legend style={draw=none},
    y label style={rotate=-90, at={(ticklabel cs:1.03)}, anchor=south west, font=\normalfont},
    x label style={at={(ticklabel cs:1.03)}, anchor=north east, font=\normalfont},
    height=5cm, width=\axisdefaultwidth,
    axis y line = left,
    axis x line = bottom,
    tick align = inside,
    axis line style = {-},  %
    axis line shift=10pt,   %
  },
  every tick/.style={
    major tick length = 0.1cm,
    minor tick length = 0.07cm,
  },
  every tick label/.style={
    font = \normalfont\rmfamily,
  },
}

\pgfplotscreateplotcyclelist{bestlist}{
  {bluefte, solid, mark = *},
  {redfte, dashed, mark = x},
  {yellowfte, dotted, mark = diamond*},
  {greenfte, dashdotted, mark = +},
  {greyfte, dashdotdotted, mark = square*},
  {purplefte, dashdotdotted, mark = square*},
}

\pgfplotscreateplotcyclelist{bestnodash}{
  {bluefte, solid, mark = *},
  {redfte, solid, mark = x},
  {yellowfte, solid, mark = diamond*},
  {greenfte, solid, mark = +},
  {greyfte, solid, mark = square*},
  {purplefte, solid, mark = asterisk},
}

\pgfplotscreateplotcyclelist{bestnomark}{
  {bluefte, solid, mark = none},
  {redfte, solid, mark = none},
  {yellowfte, solid, mark = none},
  {greenfte, solid, mark = none},
  {greyfte, solid, mark = none},
  {purplefte, solid, mark = none},
}

\pgfplotscreateplotcyclelist{shortlist}{
  {bluefte, solid, mark = o},
  {redfte, solid, mark = o},
}

\begin{document}

\title{\sysname: Redistribution-Resistant Content Protection for Decentralized Storage}

\ifdefined\Anonymous
  \author{Anonymous Author(s)}
  \institute{Anonymous Institute(s)}
\else
  \author{Giacomo Giuliari \and Karl W\"ust}
  \institute{Mysten Labs}
\fi

\maketitle

\begin{abstract}
  In decentralized storage systems, access control is often implemented by encrypting the data before
  upload and sharing the decryption key with authorized parties. A leaked key, however, makes the
  data publicly accessible, which lowers the barrier to content piracy below that of traditional
  systems, where piracy requires redistributing the full data.

  We present \sysname, an end-to-end access-control system for decentralized storage that
  secret-shares the data itself---instead of just the key---across multiple servers. Leaking the data
  then requires transmitting it in full: We formalize this intuition and introduce the
  \emph{redistribution bandwidth} an adversary must pay to leak protected content and show that
  \sysname raises it to the size of the data.

  Sharing across untrusted servers requires robustness against corrupted shares. We develop a
  robustness transform that turns any computational secret sharing scheme into a robust one and which
  is of independent interest. In contrast to existing schemes that rely on error correction, it uses
  signatures with ephemeral keys and adds only constant-size metadata per share. We show that this
  transform is secure and evaluate \sysname end-to-end on the Walrus decentralized storage system
  with an on-chain access policy.
\end{abstract}

\section{Introduction}\label{sec:intro}

Decentralized storage systems such as Filecoin~\cite{filecoin}, Arweave~\cite{williams2019arweave},
and Walrus~\cite{danezis2025walrus} promise highly available and durable storage for \emph{public
  data} and are aiming to become foundational infrastructures for data storage and distribution.

Self-sovereign content monetization by creators is often presented as one of the potential ``killer
apps'' of decentralized storage. Successful adoption, however, requires effective protection of
intellectual property. Yet, data access control---essential to intellectual property protection and
many other applications---on top of these systems has so far been overlooked. We fill this gap by
designing an access control mechanism tailored to decentralized storage systems.

A natural way to build access control on decentralized storage---e.g., for a media subscription
service (\cref{fig:key-leak-example})---is to encrypt the data and store the decryption key on a
separate server that shares it only with authenticated subscribers. While simple and efficient, this
approach suffers from two critical flaws. First, the key server is a single point of failure,
negating the benefits of decentralization. Second, and more important, the decryption key can be
easily leaked, effectively making the data publicly accessible, as anyone can then request the
encrypted data from the public storage and decrypt it locally. This lowers the barrier to piracy
significantly compared to traditional Web2 services, where content piracy requires downloading and
redistributing the full data instead of a short key. Note that, while superficially similar, this
problem is different from password sharing in subscription services (e.g., Netflix), where the
provider retains control over data access, can limit concurrent use, and passwords are generally
only shared between users who know each other. Credential sharing remains possible with any
access-control system, including ours, and is orthogonal to the leaks we address. In contrast, a
leaked decryption key with a publicly-available ciphertext enables low-effort, widespread
distribution of the data.

Secret-sharing the key as a straightforward solution, e.g., using Shamir's Secret Sharing (SSS),
across multiple key servers removes the single point of failure, but does not prevent key leakage.

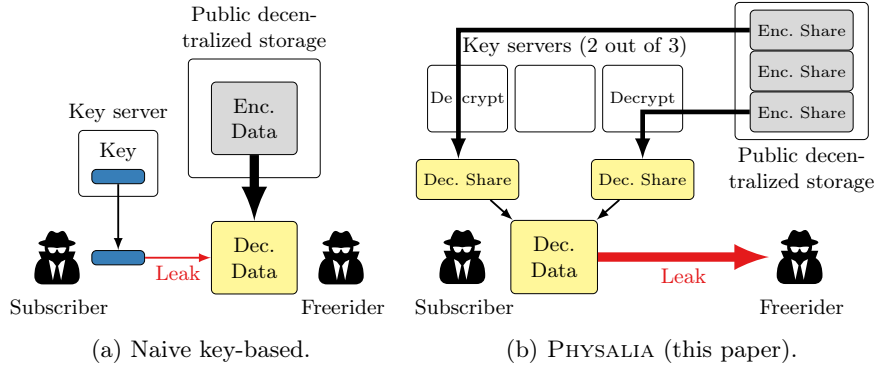
\begin{figure}[t]
  \centering
  \begin{subfigure}[t]{0.43\linewidth}
    \resizebox{\linewidth}{!}{%
      \begin{tikzpicture}[
    >=latex,
    server/.style={rectangle, draw, minimum width=1.2cm, minimum height=1cm, rounded corners=2pt},
    key/.style={rectangle, draw, fill=bluefte, rounded corners=2pt, minimum width=0.8cm, minimum
        height=0.1cm},
    encdata/.style={rectangle, draw, fill=gray!30, rounded corners=2pt, minimum width=1.3cm, minimum
        height=1.1cm, align=center, font=\footnotesize},
    decdata/.style={rectangle, draw, fill=yellow!50, rounded corners=2pt, minimum width=1.3cm,
        minimum height=1.1cm, align=center, font=\footnotesize},
    storage/.style={rectangle, draw, rounded corners=2pt, minimum width=2cm, minimum height=1.8cm},
    spy/.style={font=\Huge},
    label/.style={font=\footnotesize},
  ]

  \node[spy] (subscriber) {\faUserSecret};
  \node[label, anchor=north] at (subscriber.south) {Subscriber};

  \node[key, anchor=west] (key) at (subscriber.east) {};

  \node[key, above=of key] (ks) {};
  \node[label, anchor=south] at (ks.north) {Key};

  \node[server, fit=(ks), yshift=0.2cm] (keyserver) {};
  \node[label, anchor=south] at (keyserver.north) {Key server};

  \draw[->, thick] (ks) -- (key);

  \node[decdata, right=of key] (decdata) {Dec.\\Data};

  \draw[->, redfte, thick] (key.east) -- (decdata.west) node[midway, below,
    font=\footnotesize] {Leak};

  \node[spy, anchor=west] (other) at ($(decdata.east) + (2mm,0)$) {\faUserSecret};
  \node[label, anchor=north] at (other.south) {Freerider};

  \node[encdata, above=of decdata] (encdata) {Enc.\\Data};

  \draw[->, line width=4pt] (encdata) -- (decdata);

  \node[storage, fit=(encdata)] (storage) {};
  \node[label, anchor=south, align=center] at (storage.north) {Public decen-\\tralized storage};
\end{tikzpicture}%
    }
    \caption{Naive key-based.}
    \label{fig:key-leak-example}
  \end{subfigure}
  \begin{subfigure}[t]{0.53\linewidth}
    \resizebox{\linewidth}{!}{%
      \begin{tikzpicture}[
    >=latex,
    server/.style={rectangle, draw, minimum width=1.2cm, minimum height=1cm, rounded corners=2pt,
        font=\scriptsize},
    key/.style={rectangle, draw, fill=bluefte, rounded corners=4pt, minimum width=1cm, minimum
        height=0.3cm},
    encdata/.style={rectangle, draw, fill=gray!30, rounded corners=2pt, minimum width=1.3cm, minimum
        height=1.1cm, align=center, font=\footnotesize},
    decdata/.style={rectangle, draw, fill=yellow!50, rounded corners=2pt, minimum width=1.3cm,
        minimum height=1.1cm, align=center, font=\footnotesize},
    storage/.style={rectangle, draw, rounded corners=2pt, minimum width=2cm, minimum height=1.8cm},
    spy/.style={font=\Huge},
    label/.style={font=\footnotesize},
    share/.style={rectangle, draw, fill=gray!30, rounded corners=2pt, minimum width=1.3cm, minimum
        height=0.6cm, align=center, font=\scriptsize},
    decshare/.style={share, fill=yellow!50},
  ]

  \node[spy] (subscriber) {\faUserSecret};
  \node[label, anchor=north] at (subscriber.south) {Subscriber};
  \node[decdata, anchor=west] (decdata) at ($(subscriber.east) + (2mm,0)$) {Dec.\\Data};

  \node[decshare, above=0.3cm of decdata, xshift=-1.3cm] (share1) {Dec. Share};
  \node[server, above=0.4cm of share1] (server1) {};
  \node[font=\scriptsize]  at (server1.center) {De\enspace crypt};
  \node[server, right=0.1cm of server1] (server2) {};
  \node[server, right=0.1cm of server2] (server3) {Decrypt};
  \node[label, anchor=south, xshift=0.3cm] at (server2.north) {Key servers (2 out of 3)};

  \node[decshare] (share2) at (share1 -| server3) {Dec. Share};
  \draw[->, thick] (share1) -- (decdata);
  \draw[->, thick] (share2) -- (decdata);

  \node[share, right=of server3, yshift=-0.2cm] (share_enc1) {Enc. Share};
  \node[share, anchor=south] at (share_enc1.north) (share_enc2) {Enc. Share};
  \node[share, anchor=south] at (share_enc2.north) (share_enc3) {Enc. Share};

  \draw[->, line width=2pt] (share_enc3) -| ([xshift=-1cm]share1);
  \draw[->, line width=2pt] (share_enc1) -| (share2);

  \node[storage, fit=(share_enc1)(share_enc2)(share_enc3)] (storage) {};
  \node[label, anchor=north, align=center] at (storage.south) {Public decen-\\tralized storage};

  \node[spy] at (decdata -| share_enc1) (other) {\faUserSecret};
  \node[label, anchor=north] at (other.south) {Freerider};
  \draw[->, redfte, line width=4pt] (decdata.east) -- (other.west) node[midway, below, font=\footnotesize]
  {Leak};

\end{tikzpicture}
    }
    \caption{\sysname (this paper).}
    \label{fig:system-leak-example}
  \end{subfigure}
  \caption{Comparison of access control methods on public decentralized storage, in terms of
    fault tolerance and cost of leaking data. Arrow thickness shows communication cost. In
    both scenarios, the key servers authenticate the subscriber.}
\end{figure}

To address both flaws at once, our work proposes \sysname \footnote{Named after the
  \emph{Portuguese Man O'War}, a marine organism that is made up of a colony of genetically
  identical, but morphologically distinct, smaller organisms, similar to how the data in \sysname is
  reconstructed from smaller shares.}, an end-to-end access-control system for decentralized storage.
Unlike schemes that merely share a symmetric encryption key, \sysname distributes the data itself
across multiple servers, such that the file becomes accessible only when a threshold number of
servers cooperate. This approach prevents a one-time key leak from compromising long-term access
control, and increases the cost for any authorized reader to provide access to unauthorized
parties.

While piracy is always possible by re-sharing data, leaking a decryption key is significantly
different from leaking the data itself: A key is small compared to the data, can be shared at
almost zero cost, and is easily distributed anonymously at negligible bandwidth. Leaking the
secret-shared data, by contrast, requires reconstructing and re-hosting the full content, a burden
that grows with data size and requires significantly more effort to do at scale. We call this gap
the \emph{redistribution bandwidth} of the access-control scheme, and analyze it
in~\cref{sec:redistribution-bandwidth}.

The large redistribution bandwidth is the fundamental advantage of data-level secret sharing over
key-level secret sharing for access control in public decentralized storage systems, and the
primary motivation for \sysname.

\sysname maximizes the redistribution bandwidth required to leak protected content and forces full data
reconstruction to require substantial coordination among colluding entities. Even though our
protocol can run with servers storing data directly, it also works as a modular layer atop
decentralized storage (\cref{fig:system-leak-example}): Each server holds only a decryption key for
its share, while encrypted shares reside in public storage.

Unlike information-theoretically secure schemes such as SSS, whose shares are as large as the data,
\sysname builds on computational secret sharing (CSS), such as SSMS~\cite{krawczyk1993secret} and
AONT-RS~\cite{resch2011aontrs}, whose shares have the
optimal~\cite{aureliano2025optimalcomputationalsecretsharing} size of $\approx \frac{|D|}{k}$ for
data size $|D|$ and reconstruction threshold $k$ (\cref{sec:css}).

Sharing the data across untrusted servers additionally requires \emph{robustness}: A reader must be
able to reconstruct the data even in the presence of malicious servers that may provide invalid
shares. A reader must therefore be able to detect such invalid shares before attempting to restore
the data. Otherwise, reconstruction may fail, incorrect data may be reconstructed, or
reconstruction could become prohibitively expensive since a reader may need to perform
trial-and-error decoding for all possible combinations of shares that they received (e.g. if
integrity is checked after reconstruction).

Existing robustness mechanisms can require per-share authentication metadata that grows linearly
with $n$~\cite{bellare2007robust,chen2017revisiting}, defer integrity checks until reconstruction
and thus risk combinatorial trial decoding~\cite{bellare2020reimagining}, or rely on a
pre-authenticated dealer verification key~\cite{krawczyk1993secret} (\cref{sec:relatedwork}).

As part of \sysname's design, we therefore develop a new robustness transform that avoids these
costs: The dealer signs each share with an ephemeral, per-secret key pair, embeds the public key in
each of the shares, and deletes the signing key. Readers then use the public key that at least $k$
shares agree on, and use it to verify each share \emph{before} decoding. Thus, any $k$ valid shares
suffice and $f = \min(k-1, n-k)$ byzantine servers are tolerated. The robustness transform assumes
nothing about the underlying secret sharing scheme and thus turns \emph{any} CSS scheme into a
robust one. We call the resulting robust computational secret-sharing construction \sharingname,
the cryptographic core of \sysname. We show that it preserves secrecy and guarantees robustness,
making it of interest beyond \sysname.
In summary, \textbf{this paper makes the following contributions}:
\begin{itemize}
  \item \textbf{End-to-end system.} We design \sysname, an access control system for
        decentralized storage that works by secret-sharing the data itself. Encrypted shares reside in
        public storage, a smart contract defines the access control policy---thus keeping policy
        enforcement decentralized---and $n$ servers (of which some may be byzantine) enforce this policy to provide access to their respective shares.
  \item \textbf{Redistribution bandwidth.} We introduce the redistribution bandwidth of an
        access-controlled storage system and show that it is $\approx |D|$ for \sysname---degrading
        gracefully to $\frac{k-j}{k}|D|$ under collusion with $j$ servers---whereas any key-based
        design only requires leaks of the size of a key.
  \item \textbf{Generic robustness transform.} We develop \sharingname, a generic construction that turns any
        computational secret sharing scheme into a robust one at a constant overhead per share, tolerating $f = \min(k-1, n-k)$ corrupted shares. We prove secrecy preservation and robust correctness.
  \item \textbf{End-to-end implementation.} We implement \sysname end-to-end on Walrus with an
        on-chain access policy. Our geo-distributed experiments show that \sysname can encode and
        reconstruct at multiple tens of Gbps, supporting high-throughput protection of content.
        Further, end-to-end store latency is dominated by Walrus storage and on-chain transactions, rather
        than \sharingname. The latency overhead for writes is modest, while read latency is comparable
        to---and sometimes lower than---direct Walrus reads.
\end{itemize}

\section{\sysname System Design}\label{sec:design}

In this section, we first state our system and threat model (\cref{sec:system-model}) and then
introduce the architecture of \sysname (\cref{sec:decentralized-storage}), which ties the servers,
a smart contract, and the public storage system together. To protect the shares against byzantine
servers, it makes use of robust secret sharing, which we achieve by applying our \emph{robustness
  transform} (\cref{sec:robustness-transform}) to a computational secret sharing scheme
(\cref{sec:css}).

\subsection{System and Threat Model}\label{sec:system-model}

Our system model consists of three main entities: a \emph{writer} (e.g., a content creator or
service provider), a set of \emph{servers} that enforce access control, and a set of \emph{readers}
(authorized users). They interact through a \emph{smart contract system} that holds the access
control policy and a \emph{decentralized storage system} that stores the shares. All entities in
the model are computationally bounded.

A writer is an entity that uploads data to a decentralized storage system and wants to protect it
against unauthorized access. The writer generates shares of that data, encrypts them for the
\sysname servers, uploads the encrypted shares to the storage system, and creates an access control
policy. We assume that the writer is honest and correctly follows the protocol. This is the natural
assumption since they are usually the content owner who runs the scheme to protect their own data.
An honest dealer (the writer in our case) is also standard for robust, as opposed to verifiable,
secret sharing (cf.~\cref{sec:relatedwork}).

We consider a set of $n$ servers responsible for enforcing access control to shares of the data.
Each server stores a decryption key that allows them to decrypt their assigned shares, which reside
in decentralized storage. Note that our robustness transform (cf.~\cref{sec:robustness-transform})
could be used independently to secret-share data robustly to servers that store the shares
directly.

Of these servers, $f$ may be byzantine (i.e., behave arbitrarily), while the remainder are
honest-but-curious. The threshold $k$ that is set as a system parameter determines the number $f =
  \min(k-1, n-k)$ of byzantine servers that can be tolerated for the integrity of reconstruction.
Confidentiality holds against arbitrary coalitions of fewer than $k$ servers. For the analysis of
the redistribution bandwidth (\cref{sec:redistribution-bandwidth}), we additionally assume that a
malicious reader does not collude with the servers, i.e., no server reveals the private key that
decrypts its shares; \cref{sec:redistribution-bandwidth} quantifies the effect of relaxing this
assumption.

Finally, readers retrieve the data by fetching shares from the servers. We assume that all
non-byzantine servers enforce access control checks to ensure that only authorized readers can
access the data.

\subsection{Architecture and Workflow}\label{sec:decentralized-storage}

\pgfdeclarelayer{background}
\pgfdeclarelayer{foreground}
\pgfsetlayers{background,main,foreground}

\begin{figure*}[ht]
  \centering

  \begin{tikzpicture}[
      node distance=0.35cm and 0.6cm,
      font=\small,
      >=latex,
      data/.style={rectangle, draw, minimum width=2cm, align=center,
      rounded corners=2pt},
      algo/.style={rectangle, draw, minimum width=2cm, align=center,
      rounded corners=2pt, fill=black, text=white},
      entity/.style={rectangle, draw, minimum width=2cm, align=center,
      rounded corners=2pt, fill=bluefte!50},
      steplbl/.style={font=\large},
      dashedbox/.style={draw, dashed, rounded corners=6pt, inner sep=6pt},
      arrowunderlay/.style={line width=4pt, white, draw},
    ]

    \node[data] (fileA) {Content};
    \node[algo, below=of fileA, text width=1.8cm] (shares) {Create Shares};
    \node[algo, below=of shares] (encrypt) {Encrypt};
    \node[data, right=of encrypt.south east, anchor=south west] (accesscontrol)
    {Access\\Control\\Policy};

    \node[dashedbox, fit=(fileA) (shares) (encrypt) (accesscontrol),
    label=above:Writer] (uploader) {};

    \node[entity, below=8mm of encrypt, minimum height=0.9cm] (dstore) {Decentralized\\Storage};
    \node[entity, below=8mm of accesscontrol, minimum height=0.9cm] (scontract) {Smart\\Contract};

    \node[algo, right=4.2cm of fileA.north east, anchor=north west] (server1ac) {AC Check};
    \node[algo, below=0.2cm of server1ac] (server1dec) {Decrypt};

    \begin{pgfonlayer}{background}
      \node[dashedbox, fit=(server1ac)(server1dec), label=above:{Server 1}] (server1) {};
      \draw[arrowunderlay] ($(server1dec.south)+(-0.7,0)$) --
      ($(server1dec.south)+(-0.7cm,-3.7cm)$)
      -| ($(dstore.south)+(0.1cm,0)$);
      \draw[<-, thick] ($(server1dec.south)+(-0.7,0)$) -- ($(server1dec.south)+(-0.7cm,-3.7cm)$)
      -| ($(dstore.south)+(0.1cm,0)$);
    \end{pgfonlayer}

    \node[algo, below=1.3cm of server1] (servernac) {AC Check};
    \node[algo, below=0.2cm of servernac] (serverndec) {Decrypt};

    \begin{pgfonlayer}{background}
      \node[dashedbox, fill=white, fit=(servernac)(serverndec), label={Server $n$}] (servern) {};
    \end{pgfonlayer}

    \node[font=\huge] at ($(server1)!0.4!(servern)$) {\(\vdots\)};

    \node[algo, right=1.5cm of server1ac.north east, anchor=north west] (req) {Request\\Shares};
    \node[algo, below=of req] (recon) {Reconstruct};
    \node[data, below=of recon] (fileC) {Content};
    \node[dashedbox, fit=(req) (fileC), label=above:Reader] (readerbox) {};

    \draw[->, thick] (fileA) -- (shares);
    \draw[->, thick] ($(shares.south)-(0.6,0)$) -- ($(encrypt.north)-(0.6,0)$);
    \draw[->, thick] ($(shares.south)+(0.6,0)$) -- ($(encrypt.north)+(0.6,0)$);
    \node[font=\Large] at ($(shares)!0.5!(encrypt)$) {\dots};

    \draw[arrowunderlay] ($(encrypt.south)-(0.6,0)$) -- ($(dstore.north)-(0.6,0)$);
    \draw[->, thick] ($(encrypt.south)-(0.6,0)$) -- ($(dstore.north)-(0.6,0)$);
    \draw[arrowunderlay] ($(encrypt.south)+(0.6,0)$) -- ($(dstore.north)+(0.6,0)$);
    \draw[->, thick] ($(encrypt.south)+(0.6,0)$) -- node[steplbl, right=0.1cm] {\ding{204}}
    ($(dstore.north)+(0.6,0)$);
    \node[font=\LARGE] at ($(encrypt)!0.5!(dstore)$) {\dots};

    \draw[arrowunderlay] (accesscontrol) -- (scontract);
    \draw[->, thick] (accesscontrol) -- node[steplbl, right=0.1cm] {\ding{205}} (scontract);

    \draw[arrowunderlay] ([xshift=-0.1cm, yshift=0.1cm]req.west) -- ([xshift=0.1cm,
    yshift=0.1cm]server1.east);
    \draw[->, thick] ([xshift=-0.1cm, yshift=0.1cm]req.west) -- ([xshift=0.1cm,
    yshift=0.1cm]server1.east);
    \draw[arrowunderlay] ([xshift=-0.1cm, yshift=-0.1cm]req.west) -- ([xshift=0.1cm,
    yshift=0.1cm]servern.north east);
    \draw[->, thick] ([xshift=-0.1cm, yshift=-0.1cm]req.west) -- (servern.north east);

    \draw[<-, thick] (server1ac) -- ($(server1ac.west)-(0.8cm,0)$) |-
    ($(scontract.east)+(0cm, 0.1cm)$);
    \draw[<-, thick] (servernac) -- ($(servernac.west)-(0.6cm,0)$) |- node[steplbl,
    below left=-0.05cm] {\ding{207}}
    ($(scontract.east)+(0cm, -0.1cm)$);

    \draw[arrowunderlay] ($(serverndec.south)+(-0.5cm,0)$) --
    ($(server1dec.south)+(-0.5cm,-3.9cm)$)
    -| node[steplbl] {} ($(dstore.south)+(-0.1cm,0)$);
    \draw[<-, thick] ($(serverndec.south)+(-0.5cm,0)$) -- ($(server1dec.south)+(-0.5cm,-3.9cm)$)
    -| node[steplbl, above left=0.0cm] {\ding{208}} ($(dstore.south)+(-0.1cm,0)$);

    \draw[arrowunderlay] ([xshift=0.1cm, yshift=-0.1cm]server1dec.east) -- ([xshift=-0.1cm,
    yshift=0.1cm]recon.west);
    \draw[->, thick] ([xshift=0.1cm, yshift=-0.1cm]server1dec.east) -- ([xshift=-0.1cm,
    yshift=0.1cm]recon.west);
    \draw[arrowunderlay] ([xshift=0.1cm, yshift=-0.1cm]serverndec.east) --
    ([xshift=-0.1cm,
    yshift=-0.1cm]recon.west);
    \draw[->, thick] ([xshift=0.1cm, yshift=-0.1cm]serverndec.east) -- node[steplbl,
    below=0.2cm] {\ding{209}} ([xshift=-0.1cm,
    yshift=-0.1cm]recon.west);

    \draw[->, thick] (recon) -- (fileC);

    \node[steplbl, right=0.0cm of shares] {\ding{202}};
    \node[steplbl, right=0.0cm of encrypt] {\ding{203}};
    \node[steplbl, right=0.15cm of req] {\ding{206}};
    \node[steplbl, right=0.15cm of recon] {\ding{210}};

  \end{tikzpicture}
  \caption{\sysname architecture and workflow for access control on decentralized storage.
    \ding{202}~The writer creates shares using \sharingname, \ding{203}~encrypts each share
    with the corresponding server's public key, and \ding{204}~stores encrypted shares in decentralized
    storage. \ding{205}~An access control policy is deployed to a smart contract.
    \ding{206}~When a
    reader requests access, \ding{207}~servers verify authorization via the smart contract,
    \ding{208}~retrieve their encrypted shares from storage, \ding{209}~decrypt and provide shares
    to the reader, \ding{210}~who reconstructs the original content.}

  \label{fig:decentralized-storage}
\end{figure*}
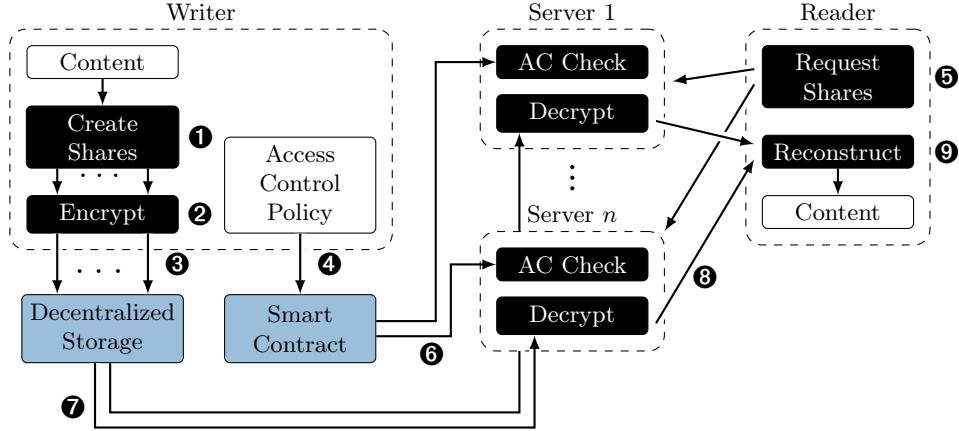

\sysname addresses the key-leak problem by secret-sharing the data itself and storing each share in encrypted form in the decentralized
storage system, where multiple servers each have access to a key that allows decryption of a
single share.
This approach ensures that (i) the data remains protected from malicious servers, and (ii) a user
needs to download and redistribute the full data to share it with third parties
(\cref{sec:redistribution-bandwidth}).
The integration of \sysname with decentralized storage systems follows the workflow illustrated in
\cref{fig:decentralized-storage}.
In the following, we sketch out the
architecture and workflow for such a system.

The writer first creates the shares using \sharingname (\cref{sec:css,sec:robustness-transform}),
acting as the dealer of the secret sharing scheme (\ding{202}~in~\cref{fig:decentralized-storage}).
They then encrypt each share using the public key of the corresponding server (\ding{203}), before
uploading each encrypted share to the decentralized storage system (\ding{204}).%

At the same time, the writer creates an access control policy (e.g., an access control list) that
defines who is authorized to access the content and deploys this policy to a smart contract
(\ding{205}). Such a policy can simply be an access control list (ACL) or it can be
programmatically defined in a smart contract.

When a reader wishes to access the content, they first authenticate themselves (depending on the
policy) to the servers and request the shares of the data (\ding{206}). The servers check if the
reader is authorized using the access control policy stored in the smart contract (\ding{207}),
and, if so, retrieve their respective encrypted share from the decentralized storage system
(\ding{208}).

Each server then decrypts the share and provides it to the reader (\ding{209}). Finally, the reader
reconstructs the original content from the shares (\ding{210}).

\subsection{Computational Secret Sharing}\label{sec:css}

Sharing the data with an information-theoretically secure scheme such as Shamir's secret
sharing~\cite{shamir1979share} would make every share as large as the data itself, for a total of
$n|D|$ stored bytes. \sysname therefore uses a \emph{computational secret sharing} scheme with
compact shares, such as SSMS~\cite{krawczyk1993secret} or AONT-RS~\cite{resch2011aontrs}. Secrecy
holds against computationally bounded adversaries, which allows shares of size $\approx
  \frac{|D|}{k}$, which is optimal~\cite{aureliano2025optimalcomputationalsecretsharing}. A $(k, n)$
CSS scheme consists of $\CSSShare(s) \rightarrow (\outsym{1}, \ldots, \outsym{n})$, which outputs
$n$ indexed shares of a secret $s$, and $\CSSRecover(T) \rightarrow s'$, which recovers the secret
from a set $T$ of indexed shares or outputs $\bot$. The scheme is \emph{complete} if any $k$ of the
shares recover $s$, and \emph{secret} if fewer than $k$ shares reveal no information about $s$ to a
computationally bounded adversary. \Cref{def:css} in \cref{sec:proofs} gives the formal definition
and \cref{def:css-ind} the corresponding game; \cref{sec:instantiations} describes two
instantiations.

\subsection{The Robustness Transform}\label{sec:robustness-transform}

Standard CSS does not suffice in the presence of byzantine servers. A malicious server can return a
modified share, and the reader cannot distinguish it from a correct one before decoding: With
corrupted shares among the responses, the reader may reconstruct incorrect data, or they have to
attempt decoding for all combinations of the received shares until a consistent subset is found,
which quickly becomes prohibitively expensive.

Existing robust schemes fall short as we discuss in~\cref{sec:relatedwork}. We therefore take a
different approach and add robustness with a novel robustness transform that uses signatures with
ephemeral, per-secret keys. We use a signature scheme that is existentially unforgeable under
chosen-message attacks (EUF-CMA), writing $\sign{sk}{m}$ for signing and $\verify{pk}{m}{\sigma}$
for verification. The robustness transform assumes nothing about the CSS scheme beyond
\cref{sec:css} and thus lifts \emph{any} CSS scheme to a robust one, which makes it of interest
beyond \sysname. It consists of two phases, \emph{Share Generation} and \emph{Reconstruction}.

A \emph{robust} CSS scheme $(\RSSShare, \RSSReconstruct)$ provides \emph{robust correctness} in
addition to the properties of a standard CSS scheme: If at most $f = \min(k-1, n-k)$ shares are
corrupted and at least $k$ original shares remain, $\RSSReconstruct$ outputs the original secret,
except with negligible probability. This is formalized in \cref{def:rss,def:rss-malleability} in
\cref{sec:proofs}.

\begin{figure*}[t]
  \centering
  \def\vdist{0.2cm} %
  \begin{tikzpicture}[%
      node distance=0.6cm,
      >=latex,
      data/.style={rectangle, draw, minimum height=3em, minimum width=4em, align=center,
      rounded corners=2pt},
      algo/.style={rectangle, draw, minimum height=3em, minimum width=4em, align=center,
      rounded corners=2pt, fill=black, text=white},
      share/.style={rectangle, draw, minimum height=1.70em, minimum width=4em,
      align=center, rounded corners=2pt, bluefte},
      secret/.style={rectangle, draw, minimum width=4em, align=center,
      rounded corners=2pt, fill=redfte},
      sharecontent/.style={rectangle, draw, minimum height=1.5em, minimum width=4em, align=center,
      rounded corners=2pt},
    ]

    \node[algo] (ec) at (0,0) {$\CSSShare$};
    \node[algo, above={2.5*\vdist} of ec, minimum height=1.5em] (sign) {Sign};
    \node[secret, left=of sign, minimum height=1.5em] (sk) {Ephemeral SK};
    \node[algo, above=\vdist of sign, minimum height=1.5em] (derive) {Derive};

    \node[fit=(ec)(sk)(derive), rounded corners=2pt, draw, inner sep=5pt,
      dashed, purplefte,
    thick, label=above:{\textcolor{purplefte}{\textbf{\sharingname}}}] (physalia) {};

    \node[data, anchor=east] (data) at ($(physalia.west |- ec)+(-0.6cm,0)$) {Data};

    \node[sharecontent, right=1.0cm of derive] (pk_zoom1) {PK};
    \node[sharecontent] (sig_zoom1) at (pk_zoom1 |- sign) {Sig. $1$};
    \node[sharecontent, below=\vdist of sig_zoom1] (idx_zoom1) {$1$};
    \node[sharecontent, below=\vdist of idx_zoom1] (sym_zoom1) {Share $1$};

    \node[fit=(idx_zoom1)(sym_zoom1), inner sep = 2pt, rounded corners=2pt, draw,
    dotted, thick] (index_symbol_box1) {};

    \node[fit=(index_symbol_box1)(sig_zoom1)(pk_zoom1), rounded corners=2pt, inner sep=4pt, draw,
    thick, share, label=above:{\textcolor{bluefte}{\textbf{Robust Share $1$}}}] (share1_box) {};

    \node[sharecontent, right=1.6cm of pk_zoom1] (pk_zoom_n) {PK};
    \node[sharecontent] (sig_zoom_n) at (pk_zoom_n |- sign) {Sig. $n$};
    \node[sharecontent, below=\vdist of sig_zoom_n] (idx_zoom_n) {$n$};
    \node[sharecontent, below=\vdist of idx_zoom_n] (sym_zoom_n) {Share $n$};

    \node[fit=(idx_zoom_n)(sym_zoom_n), inner sep = 2pt, rounded corners=2pt, draw,
    dotted, thick] (index_symbol_box_n) {};

    \node[fit=(index_symbol_box_n)(sig_zoom_n)(pk_zoom_n), rounded corners=2pt, inner
      sep=4pt, draw,
    thick, share, label=above:{\textcolor{bluefte}{\textbf{Robust Share $n$}}}] (share_n_box) {};

    \node[font=\LARGE, text=bluefte] at ($(share1_box)!0.5!(share_n_box)$) {\dots};

    \draw[white, line width=4pt] (data) -- (ec);
    \draw[->] (data) -- (ec);
    \draw[line width=4pt, white] (ec.east |- sym_zoom1.west) -- (sym_zoom1);
    \draw[->] (ec.east |- sym_zoom1.west) -- (sym_zoom1);
    \draw[->] (sk) -- (sign);
    \draw[line width=4pt, white] (sign) -- (sig_zoom1);
    \draw[->] (sign) -- (sig_zoom1);
    \draw[->] (sk) |- (derive);
    \draw[line width=4pt, white] (derive) -- (pk_zoom1);
    \draw[->] (derive) -- (pk_zoom1);
    \draw[line width=4pt, white] (index_symbol_box1.west) --
    ($(index_symbol_box1.west)+(-0.25cm,0cm)$) --
    ($(index_symbol_box1.west)+(-0.25cm,0.5cm)$) -| (sign);
    \draw[->] (index_symbol_box1.west) -- ($(index_symbol_box1.west)+(-0.25cm,0cm)$) --
    ($(index_symbol_box1.west)+(-0.25cm,0.5cm)$) -| (sign);

  \end{tikzpicture}
  \caption{Overview of the share generation in \sharingname. The underlying CSS scheme produces
    the shares, each of which the
    robustness transform signs with the randomly-generated ephemeral secret key, which is the same
    for all shares. The boxes on the right show a zoomed in view of robust shares 1 and $n$,
    illustrating their components: Public key, signature, index, and CSS share. The signature is
    over the index and CSS share.}
  \label{fig:share-encoding}
\end{figure*}
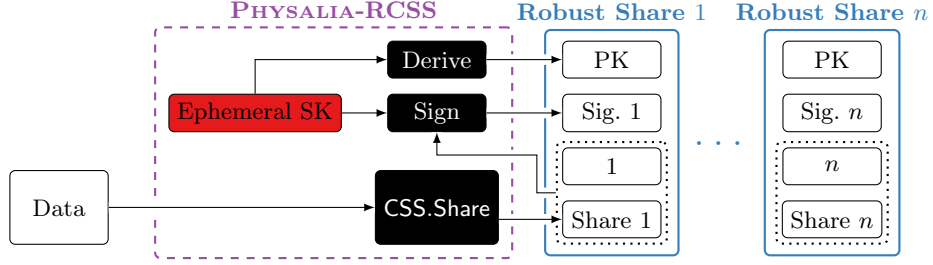

\paragraph{Share Generation.}\label{sec:share-generation-and-distribution}
The share generation phase transforms a secret $s$ (e.g., a large data blob) into $n$ \emph{robust
  shares} and distributes them to different servers. This phase, i.e., the algorithm $\RSSShare$
(shown in \cref{alg:share} in \cref{sec:pseudocode}), is illustrated in \cref{fig:share-encoding}
and consists of the following steps, executed by the dealer: The secret $s$ is first shared using
the underlying CSS scheme, producing $n$ indexed shares $(\outsym{1}, \ldots, \outsym{n}) \gets
  \CSSShare(s)$, any $k$ of which suffice to recover $s$. The dealer then generates a random
ephemeral public-private key pair $(sk, pk)$. For each share $\outsym{i}$ created in the previous
step, they sign the unsigned share $(i,\outsym{i})$ using the ephemeral private key to obtain a
signature $\sig{i} \gets \sign{sk}{(i,\outsym{i})}$. Robust share $i$ is then the tuple $\share{i}
  \gets (pk, \sig{i}, i, \outsym{i})$. After this step, the dealer deletes the ephemeral private key
$sk$ , and outputs the list of robust shares $[\share{1}, \ldots, \share{n}]$, sending each share
to the corresponding server.

Because the key pair is freshly generated for each shared secret, the private key is deleted after
sharing, and the public key travels inside the shares, no public-key infrastructure and no
long-term dealer identity are required. The verification key is established by the agreement of the
shares themselves.

\paragraph{Reconstruction.}\label{sec:reconstruction-phase}
To reconstruct the secret, a reader first requests shares from servers and waits for at least $k$
shares to be received. The reconstruction phase, i.e., the algorithm $\RSSReconstruct$ (shown in
\cref{alg:reconstruct} in \cref{sec:pseudocode}), takes as input a set of shares $S$ and outputs
either a secret $s'$ or $\bot$. If $\RSSReconstruct$ outputs $\bot$, the set $S$ does not contain a
sufficient number of valid shares. In this case, the reader requests additional shares and retries;
since at most $f = \min(k-1, n-k)$ servers are byzantine, this eventually succeeds.
$\RSSReconstruct$ consists of the following steps:

\paragraph{Public Key Agreement.} The reader checks if at least $k$ shares in the set $S$ contain the same public key $pk$. If this
is not the case, the reader outputs $\bot$. Otherwise, the reader uses this public key $pk$ to
verify the shares in the following step. Since at least $k$ shares are from honest servers and at
most $f < k$ are corrupted, this public key corresponds to the ephemeral key used by the dealer.

\paragraph{Share Verification.} For each share $\share{i} = (pk_i, \sig{i}, i, \outsym{i})$ in $S$, the reader checks that (i)
$pk_i = pk$ and (ii) $\verify{pk}{(i, \outsym{i})}{\sig{i}} = \mathrm{true}$, i.e. the public key
of the share matches the correct key and the signature is valid. If either check fails, the reader
discards the share. If fewer than $k$ shares remain, the reader outputs $\bot$.

\paragraph{CSS Recovery.} Once at least $k$ valid shares have been collected, the reader recovers the secret by running the
recovery algorithm of the underlying CSS scheme on the verified indexed shares, i.e., $s \gets
  \CSSRecover(T)$ for $T = \{(i, \outsym{i}) \mid \share{i} \text{ valid}\}$. By the completeness of
the CSS scheme (\cref{sec:css}), the recovery succeeds whenever the verified shares are authentic.

\section{Redistribution Bandwidth}\label{sec:redistribution-bandwidth}

When deployed as an access-control layer on top of a public decentralized storage system, \sysname
forces an adversary who wishes to leak the data to exfiltrate information proportional to the full
data size, rather than a short cryptographic key. We call this concept the \emph{redistribution
  bandwidth}, and further explain it by comparing \sysname (\cref{sec:decentralized-storage}) with
the natural alternative in which the data $D$ is encrypted before storing it in a public
decentralized storage system and only the decryption key is secret-shared.

\paragraph{Threat Model.}
We focus here on the \emph{reader} side of the access-control protocol: A reader $\mathcal{R}$ who
has legitimately read $D$, turns malicious and wants to provide access to $D$ to other illegitimate
readers while minimizing the amount of data that they need to share. We assume non-collusion
between $\mathcal{R}$ and the servers; in particular, servers do not reveal their private keys that
allow decryption of their respective shares (in the public decentralized storage) to $\mathcal{R}$.
We relax this assumption below, where we quantify the effect of collusion with $j < k$ servers. The
sharing of authentication credentials (e.g., a streaming service password) is orthogonal to this
analysis: It affects any access-control system, can be mitigated by standard rate-limiting and
device management, and does not enable the low-cost, large-scale distribution we aim to prevent.

\paragraph{Setup.} We consider a publicly readable decentralized storage system and a set of servers
that collectively enforce the access control policy. In \sysname (\cref{sec:decentralized-storage}),
the writer produces $n$ shares, encrypts each share $s_i$ under the corresponding server's public
key $pk_i$, and publishes the encrypted shares in the decentralized storage system. An authorized
reader retrieves $k$ decrypted shares from $k$ servers and runs reconstruction. In the \emph{naive
deployment}, the writer instead encrypts the data with a key $K$, publishes the ciphertext $E(K, D)$
on the storage layer and secret-shares the key $K$ across the same servers; an authorized reader
reconstructs $K$ and decrypts the public ciphertext.

A legitimate reader $\mathcal{R}$ has been authorized, retrieved their shares (of the data or the
key, respectively), and reconstructed $D$ or $K$, respectively. Now $\mathcal{R}$ wishes to enable a
non-authorized party $\mathcal{U}$ to read $D$ with the minimum required egress cost of $\redbw$
bits. In the naive deployment, $\redbw = |K|$ bits suffice: $\mathcal{R}$ transmits $K$, and
$\mathcal{U}$ fetches $E(K, D)$ from public storage and decrypts. In \sysname, the public storage
holds only shares encrypted under keys that $\mathcal{R}$ does not have, so $\mathcal{R}$ must
transmit essentially the full data, $\redbw \approx |D|$. The redistribution bandwidth ratio between
the two deployments is therefore $|D|/|K|$: The size of the protected content versus a single
cryptographic key. This gap does not depend on how the key is managed in the naive deployment:
Whether $K$ is held by one server, secret-shared, encapsulated with attribute-based
encryption~\cite{goyal2006attribute}, or released by threshold key-management systems such as
Seal~\cite{mystenlabs2025seal}, the public storage holds a ciphertext that any key holder can
unlock.

\paragraph{Graceful Degradation under Collusion.}
Reader--server collusion weakens the guarantee only gradually: If $j < k$ servers leak their
decryption keys, an outsider still needs the remaining fraction $\frac{k-j}{k}|D|$ of the data from
the reader. The guarantee thus degrades linearly in the number of colluding servers rather than
collapsing, and only a coalition of $k$ servers eliminates the redistribution cost entirely.

\section{Security Guarantees}\label{sec:security}

We argue informally that the robustness transform turns any secure CSS scheme (\cref{sec:css}) into
a robust secret sharing scheme; the security games, formal theorem statements, and proofs are
provided in \cref{sec:proofs}. Our analysis assumes an honest dealer (\cref{sec:system-model}) and
consists of two guarantees that hold for \emph{any} underlying CSS scheme: The transform preserves
the secrecy of the underlying scheme (\cref{thm:secrecy-preservation}), and it provides robust
correctness, by reduction to the unforgeability of the signature scheme
(\cref{thm:robust-correctness}).

\paragraph{Secrecy Preservation.} If the underlying CSS scheme is secure, so is the transformed
scheme since the transform does not introduce any additional secret-dependent values. The only
values added to a share are the ephemeral public key, which is sampled independently of the secret,
and a signature that depends only on the single share it accompanies. Any adversary against the
transformed scheme can therefore be used to construct an adversary against the underlying scheme
with the same advantage, by generating the ephemeral key pair and signing the queried shares using
the ephemeral key (\cref{thm:secrecy-preservation}).

\paragraph{Robust Correctness.} If the signature scheme is EUF-CMA-secure, at most $f = \min(k-1,
n-k)$ shares are corrupted, and at least $k$ original shares remain, reconstruction outputs the
original secret (\cref{thm:robust-correctness}). Since fewer than $k$ shares are corrupted, the
public key on which at least $k$ shares agree is the dealer's ephemeral key, and every verified
share is either original or carries a forged signature under that key. Both ways of defeating the
scheme thus require a forgery: Reconstructing a different secret requires a modified share among the
$k$ verified ones, and rejecting despite $k$ original shares contradicts the completeness of the CSS
scheme unless a modified share was accepted.

\section{Evaluation}\label{sec:evaluation}

We measure the cost of applying our robustness transform to existing CSS schemes, and compare it to
existing work. Then, we fully implement \sysname, using Walrus for decentralized storage and the Sui
blockchain as smart-contract platform, and evaluate end-to-end store and read latency.
Reproducibility details and extended results appear in \cref{app:evaluation}.

\subsection{Microbenchmarks}\label{sec:microbenchmarks}
\paragraph{Setup.}
We implement \sharingname in Rust with two underlying CSS schemes, SSMS and AONT-RS
(\cref{sec:instantiations}; \sysname{}-SSMS and \sysname{}-AONT in the figures), using
AES-128-CTR~\cite{rustctr}, BLAKE3~\cite{rustblake3}, Reed--Solomon (RS) dispersal through
\texttt{reed-solomon-simd}~\cite{rustreedsolomon}, Rayon~\cite{rustrayon} for parallelization, and
ephemeral Ed25519~\cite{bernstein2012highspeed} signatures that are computed over the BLAKE3 hash
of the share and its metadata for efficiency. We conduct two separate experiments: one measures
encoding and reconstruction throughput; the other separates base sharing time from the overhead of
the robustness transform. We compare our SSMS- and AONT-RS-based instantiations with (i) HK1,
Krawczyk's fingerprint-based robust SSMS~\cite{krawczyk1993secret,bellare2007robust}, (ii) an
HK2-style variant~\cite{bellare2007robust}, and (iii) RAONT-RS~\cite{chen2017revisiting}. In
contrast to our transform, all three may have a per-share overhead linear in $n$
(see~\cref{sec:relatedwork}).

We implement and optimize all algorithms, reusing common primitive implementations and
parallelization infrastructure wherever applicable, to ensure a fair comparison. Both experiments
run on an Apple M2 Max laptop and span $(n,k)=(3,2),(5,3),(10,6),(10,9)$, 1 and 12 parallel workers
(threads executing the local sharing computation), and plaintext sizes of 1, 10, 100, and
1,000\,MB.

\begin{figure}[!t]
  \centering
  \input{figures/micro/throughput_10_6}
  \caption{Encoding and reconstruction throughput for robustness schemes with $(n,k)=(10,6)$
    and 12 workers.}
  \label{fig:micro-throughput}
  \par\medskip
  \input{figures/micro/breakdown_10_6}
  \caption{Base sharing time and robustness overhead at 1\,GB, $(10,6)$, and 12 workers.
    The two \sysname bars use SSMS and AONT-RS, respectively.}
  \label{fig:micro-breakdown}
\end{figure}

Encoding produces $n$ shares; reconstruction receives the first or last $k$ uncorrupted shares in
index order. Our RS code is systematic, i.e., the first $k$ shares carry the dispersed data
verbatim, so first-$k$ reconstruction needs no erasure decoding, whereas last-$k$ must recover $n-k$
systematic shares: These are the best and worst cases for the dispersal layer. Both paths still
perform robustness verification and decryption. Throughput uses 20 Criterion~\cite{rustcriterion}
samples per case. The cost-breakdown experiment uses 20 paired samples, measuring total and base
time within the same invocation. Timings cover complete sharing operations, excluding networking and
access control. \Cref{app:micro-methodology} details the
hardware, timing boundaries, and HK2 variant.

\paragraph{Throughput results.}
\Cref{fig:micro-throughput} illustrates $(n,k)=(10,6)$ with 12 workers. At 1\,GB, \sysname{}-SSMS
reaches 48.3, 77.6, and 33.8\,Gbps for encoding, first-$k$, and last-$k$ reconstruction,
respectively, versus 29.8, 41.7, and 25.7\,Gbps for \sysname{}-AONT. The SSMS-based instantiation
has higher throughput across 32/32 encoding, 29/32 first-$k$, and 28/32 last-$k$ cases in the full
matrix. Encoding throughput is 1.32--1.75 times the AONT-based rate; reconstruction reversals occur
with 12 workers at 1\,MB, plus last-$k$ at 10\,MB with $(10,9)$.

AONT-RS typically has lower throughput because it additionally hashes the full ciphertext during both
encoding and reconstruction, a cost SSMS avoids by secret-sharing the short encryption key.
Regarding robustness, HK1 and HK2 exceed \sysname{}-SSMS throughput in 71 and 73 of the 96 matched
cases. This is due to the additional cost of signing each share, which in return decouples
authentication from the sharing parameters: Metadata per share is constant, each share is verified
in isolation, and any $k$ valid shares suffice (cf.~\cref{sec:relatedwork}). \Cref{app:micro-matrix}
plots all four $(n,k)$ configurations with 12 workers.

\paragraph{Robustness cost results.}
SSMS is the faster base construction in the cost-breakdown experiment. \Cref{fig:micro-breakdown}
separates base time from robustness overhead at 1\,GB, $(10,6)$, and 12 workers. Here,
\sysname{}-SSMS has reconstruction latency within 4.2\% of HK1/HK2, with an 18--19\% encoding
premium. To compare robustness on the same base, \sysname{}-AONT incurs 5\% less robustness
overhead than RAONT-RS for encoding, and 11\% and 2\% more for first-$k$ and last-$k$
reconstruction, respectively. Thus, robustness costs are comparable, while the faster SSMS base
generally improves overall performance. Cost breakdowns for all four $(n,k)$ configurations at
1\,GB with 12 workers appear in \cref{app:micro-matrix}.

\subsection{End-to-End Experiments}\label{sec:macro-experiments} \paragraph{Methodology.} We
implement \sysname end to end with the SSMS-based instantiation, using the Sui
blockchain~\cite{blackshear2023sui} for access policies and Walrus~\cite{danezis2025walrus} as the
decentralized storage system. Each share is encrypted for its server using ECIES-style hybrid
encryption: ephemeral--static X25519 key agreement, BLAKE3 key derivation, and AES-256-GCM.
We evaluate this implementation on their testnets, against a baseline
that stores each plaintext payload as one Walrus blob without any access control. Both \sysname and
the baseline contact Walrus storage nodes directly. The experimental setup is geo-distributed, with
client and servers divided across Europe and the US; deployment details appear in
\cref{app:macro-methodology}. We use the same four $(n,k)$ configurations and payload sizes as
above. We distinguish \emph{greedy} reads, which request all $n$ shares concurrently and cancel
outstanding requests once reconstruction succeeds, from \emph{lazy} reads, which initially contact
the first $k$ configured servers and request more if a request or reconstruction fails. For each
payload size and $(n,k)$ configuration, we perform ten stores and ten reads per strategy; the
baseline has ten stores and ten reads per size. Reads repeatedly access fixed objects, including the
first read. We report successful-attempt latencies from a sequential, nonrandomized run.

\begin{figure}[!t]
  \centering
  \begin{tikzpicture}
	\begin{groupplot}[
			group style={group size=2 by 1,horizontal sep=1.2cm},
			width=0.45\linewidth, height=4.1cm,
			xlabel={Data Size},
			xmin=1,
			xmax=1000,
			xmode=log, log basis x=10,
			xtick={1,10,100,1000},
			xticklabels={1MB,10MB,100MB,1GB},
			reverse legend,
			legend style={at={(0,1.02)},anchor=north west,legend columns=1,
					font=\footnotesize,legend cell align=left,draw=none,fill=none},
		]
		\nextgroupplot[ylabel={Store latency (s)},ymin=0,ymax=100,ytick={0,20,40,60,80,100}]
		\addplot[only marks,mark=*,mark size=0.8pt,color=black,opacity=0.3,mark options={solid,fill=black,draw=none},forget plot] coordinates {
				(0.922571427, 17.024649572)
				(0.939242591, 14.584686601)
				(0.956215008, 14.566622477)
				(0.973494121, 13.339752665)
				(0.991085475, 17.823388017)
				(1.008994709, 18.56626489)
				(1.027227569, 21.734326461)
				(1.045789903, 19.642566429)
				(1.064687664, 14.828885022)
				(1.083926914, 34.580197231)
				(9.225714272, 19.709636411)
				(9.392425907, 16.793931479)
				(9.562150076, 13.860573208)
				(9.734941215, 17.854948615)
				(9.910854745, 12.525759518)
				(10.08994709, 13.411086131)
				(10.272275692, 20.916330395)
				(10.45789903, 19.043895301)
				(10.646876642, 13.76552039)
				(10.83926914, 17.669564846)
				(92.257142715, 29.246441513)
				(93.92425907, 23.343560438)
				(95.621500755, 17.885472771)
				(97.349412146, 15.783537764)
				(99.108547453, 20.603647981)
				(100.899470903, 14.776187398)
				(102.722756918, 17.705736581)
				(104.5789903, 17.055078509)
				(106.468766419, 18.052440782)
				(108.392691402, 17.424016872)
				(922.571427155, 32.999344844)
				(939.2425907, 35.288406432)
				(956.215007553, 38.412333303)
				(973.494121458, 41.412993255)
				(991.085474531, 33.898662765)
				(1008.994709032, 42.790341412)
				(1027.227569183, 33.389064531)
				(1045.789903004, 42.63207302)
				(1064.68766419, 33.221692068)
				(1083.926914021, 34.312750798)
			};
		\addplot[only marks,mark=*,mark size=0.8pt,color=bluefte,opacity=0.3,mark options={solid,fill=bluefte,draw=none},forget plot] coordinates {
				(0.922571427, 29.919555447)
				(0.939242591, 35.120765528)
				(0.956215008, 21.40705891)
				(0.973494121, 20.775813902)
				(0.991085475, 22.628263177)
				(1.008994709, 28.642121612)
				(1.027227569, 19.739525766)
				(1.045789903, 22.50193444)
				(1.064687664, 28.680321645)
				(1.083926914, 31.919091413)
				(9.225714272, 20.735619581)
				(9.392425907, 30.705729968)
				(9.562150076, 21.91548686)
				(9.734941215, 27.036956159)
				(9.910854745, 23.123405053)
				(10.08994709, 19.726884239)
				(10.272275692, 17.308712362)
				(10.45789903, 20.93836984)
				(10.646876642, 23.7348435)
				(10.83926914, 19.988383208)
				(92.257142715, 20.692925831)
				(93.92425907, 18.405411265)
				(95.621500755, 25.991053531)
				(97.349412146, 20.670581756)
				(99.108547453, 20.454332461)
				(100.899470903, 22.454184507)
				(102.722756918, 28.91348746)
				(104.5789903, 22.343087808)
				(106.468766419, 20.075316291)
				(108.392691402, 31.30856892)
				(922.571427155, 62.585122202)
				(939.2425907, 48.506681781)
				(956.215007553, 49.182245429)
				(973.494121458, 85.435040158)
				(991.085474531, 66.968061563)
				(1008.994709032, 52.91084609)
				(1027.227569183, 61.531482102)
				(1045.789903004, 49.612418518)
				(1064.68766419, 56.123725931)
				(1083.926914021, 51.980255488)
			};
		\addplot[black,solid,mark=diamond*,mark size=2pt] coordinates {
				(1, 18.669133936)
				(10, 16.555124629)
				(100, 19.187612061)
				(1000, 36.835766243)
			};
		\addlegendentry{Walrus}
		\addplot[bluefte,solid,mark=*,mark size=2pt] coordinates {
				(1, 26.133445184)
				(10, 22.521439077)
				(100, 23.130894983)
				(1000, 58.483587926)
			};
		\addlegendentry{\sysname}
		\nextgroupplot[ylabel={Read latency (s)},ymode=log,log ticks with fixed point,
			ymin=1,ymax=20,ytick={1,2,5,10,20}]
		\addplot[only marks,mark=*,mark size=0.8pt,color=black,opacity=0.3,mark options={solid,fill=black,draw=none},forget plot] coordinates {
				(0.922571427, 1.608719286)
				(0.939242591, 1.465004677)
				(0.956215008, 1.786174741)
				(0.973494121, 1.872176304)
				(0.991085475, 2.490557299)
				(1.008994709, 1.572262716)
				(1.027227569, 2.660402907)
				(1.045789903, 1.954597142)
				(1.064687664, 1.527511697)
				(1.083926914, 1.913126915)
				(9.225714272, 1.679312555)
				(9.392425907, 1.791482438)
				(9.562150076, 2.097043025)
				(9.734941215, 1.972012999)
				(9.910854745, 1.83352801)
				(10.08994709, 2.101530234)
				(10.272275692, 1.654305065)
				(10.45789903, 1.843703654)
				(10.646876642, 2.34461606)
				(10.83926914, 1.708191666)
				(92.257142715, 3.313611704)
				(93.92425907, 3.218200437)
				(95.621500755, 3.323665989)
				(97.349412146, 3.208012122)
				(99.108547453, 3.45589814)
				(100.899470903, 3.095995914)
				(102.722756918, 3.635941236)
				(104.5789903, 3.221037524)
				(106.468766419, 3.209178554)
				(108.392691402, 3.589151605)
				(922.571427155, 12.79867308)
				(939.2425907, 13.500219112)
				(956.215007553, 13.101835981)
				(973.494121458, 13.268147948)
				(991.085474531, 13.411779028)
				(1008.994709032, 13.354541237)
				(1027.227569183, 12.703828761)
				(1045.789903004, 12.857116121)
				(1064.68766419, 12.92540861)
				(1083.926914021, 12.90637209)
			};
		\addplot[only marks,mark=*,mark size=0.8pt,color=bluefte,opacity=0.3,mark options={solid,fill=bluefte,draw=none},forget plot] coordinates {
				(0.922571427, 1.600139079)
				(0.939242591, 1.607870199)
				(0.956215008, 1.492527556)
				(0.973494121, 1.493915709)
				(0.991085475, 1.701726352)
				(1.008994709, 1.548391866)
				(1.027227569, 1.716437413)
				(1.045789903, 1.44500271)
				(1.064687664, 1.536045416)
				(1.083926914, 1.548757579)
				(9.225714272, 1.847038354)
				(9.392425907, 1.979625493)
				(9.562150076, 1.83028003)
				(9.734941215, 1.871148918)
				(9.910854745, 1.81286042)
				(10.08994709, 2.036719125)
				(10.272275692, 2.149622584)
				(10.45789903, 1.774079653)
				(10.646876642, 1.819867876)
				(10.83926914, 1.967797926)
				(92.257142715, 3.306708775)
				(93.92425907, 3.34903049)
				(95.621500755, 3.378906579)
				(97.349412146, 3.135719692)
				(99.108547453, 2.870119741)
				(100.899470903, 3.402286568)
				(102.722756918, 3.211285226)
				(104.5789903, 3.032191299)
				(106.468766419, 3.024620418)
				(108.392691402, 3.031046589)
				(922.571427155, 12.150677478)
				(939.2425907, 10.757893432)
				(956.215007553, 10.521929949)
				(973.494121458, 11.011462824)
				(991.085474531, 9.346128237)
				(1008.994709032, 9.898547197)
				(1027.227569183, 8.960057404)
				(1045.789903004, 8.983081862)
				(1064.68766419, 9.735622084)
				(1083.926914021, 8.952387314)
			};
		\addplot[only marks,mark=*,mark size=0.8pt,color=bluefte,opacity=0.3,mark options={solid,fill=bluefte,draw=none},forget plot] coordinates {
				(0.922571427, 2.285014174)
				(0.939242591, 2.102681513)
				(0.956215008, 2.120872613)
				(0.973494121, 2.178129104)
				(0.991085475, 1.809915589)
				(1.008994709, 2.18638376)
				(1.027227569, 2.130621527)
				(1.045789903, 2.123498626)
				(1.064687664, 1.784307748)
				(1.083926914, 1.854318855)
				(9.225714272, 2.65705656)
				(9.392425907, 3.026978114)
				(9.562150076, 3.006078431)
				(9.734941215, 2.375896803)
				(9.910854745, 2.854763222)
				(10.08994709, 2.544538849)
				(10.272275692, 3.046319135)
				(10.45789903, 2.795798343)
				(10.646876642, 2.73289666)
				(10.83926914, 2.380213261)
				(92.257142715, 4.505552842)
				(93.92425907, 4.578896246)
				(95.621500755, 4.185881399)
				(97.349412146, 4.04195568)
				(99.108547453, 4.616269947)
				(100.899470903, 5.11413553)
				(102.722756918, 4.232398035)
				(104.5789903, 4.4005213)
				(106.468766419, 4.05472136)
				(108.392691402, 3.790843676)
				(922.571427155, 13.468697734)
				(939.2425907, 13.112935555)
				(956.215007553, 12.684394358)
				(973.494121458, 12.786009623)
				(991.085474531, 12.056632454)
				(1008.994709032, 12.595222706)
				(1027.227569183, 12.424741435)
				(1045.789903004, 12.596269708)
				(1064.68766419, 12.715540834)
				(1083.926914021, 12.931896723)
			};
		\addplot[black,solid,mark=diamond*,mark size=2pt] coordinates {
				(1, 1.885053368)
				(10, 1.902572571)
				(100, 3.327069322)
				(1000, 13.082792197)
			};
		\addlegendentry{Walrus}
		\addplot[bluefte,solid,mark=*,mark size=2pt] coordinates {
				(1, 1.569081388)
				(10, 1.908904038)
				(100, 3.174191538)
				(1000, 10.031778778)
			};
		\addlegendentry{\sysname{}-greedy}
		\addplot[bluefte,dashed,mark=square*,mark size=2pt] coordinates {
				(1, 2.057574351)
				(10, 2.742053938)
				(100, 4.352117602)
				(1000, 12.737234113)
			};
		\addlegendentry{\sysname{}-lazy}
	\end{groupplot}
\end{tikzpicture}
  \caption{Store and read latency with $(n,k)=(10,6)$. Lines show means of ten successful
    attempts; faint dots are individual observations. Greedy reads request all shares concurrently;
    lazy reads initially request $k$. The read axis is logarithmic.}
  \label{fig:macro-latency}
  \par\medskip
  \begin{tikzpicture}
  \begin{groupplot}[
    group style={group size=2 by 1,horizontal sep=0.9cm},
    width=0.425\linewidth, height=4.5cm,
    ybar stacked, /pgf/bar width=11pt, /pgf/bar shift=0pt,
    reverse legend,
    ylabel style={yshift=11pt},
    ymin=0, ymax=65,
    minor ytick={65},
    xmin=-0.5, xmax=4.5, clip=false,
    x axis line style={shorten <={0.3*\pgfplotspointxaxislength},shorten >={0.1*\pgfplotspointxaxislength}},
    xtick={-0.3,1,2,3,4},
    xticklabels={Walrus,{(3,2)},{(5,3)},{(10,6)},{(10,9)}},
    x tick label style={align=center,font=\scriptsize},
    xlabel={\sysname},
    xlabel style={at={(xticklabel cs:0.6)},anchor=north,yshift=5pt,font=\scriptsize},
    legend style={at={(1.05,1)},anchor=north west,legend columns=1,
      name=breakdown-legend,font=\footnotesize,legend cell align=left,draw=none},
  ]
  \nextgroupplot[ylabel={Store latency (s) - \textbf{1MB}}]
  \addplot[fill=greenfte,draw=white,line width=0.3pt,forget plot] coordinates {
    (-0.3, 18.129120901)
    (1, 14.389209772)
    (2, 16.453092596)
    (3, 17.88242633)
    (4, 17.830679421)
  };
  \addplot[fill=greyfte,draw=white,line width=0.3pt,forget plot] coordinates {
    (-0.3, 0.5400124)
    (1, 0.41873475)
    (2, 0.382010615)
    (3, 0.364010096)
    (4, 0.438444093)
  };
  \addplot[fill=bluefte,draw=white,line width=0.3pt,forget plot] coordinates {
    (-0.3, 0)
    (1, 0.010920125)
    (2, 0.011160191)
    (3, 0.011152816)
    (4, 0.010352264)
  };
  \addplot[fill=yellowfte,draw=white,line width=0.3pt,forget plot] coordinates {
    (-0.3, 0)
    (1, 4.606932831)
    (2, 3.356128275)
    (3, 4.945270556)
    (4, 6.036685189)
  };
  \addplot[fill=purplefte,draw=white,line width=0.3pt,forget plot] coordinates {
    (-0.3, 0)
    (1, 0.00212769)
    (2, 0.002322629)
    (3, 0.002508342)
    (4, 0.001933298)
  };
  \addplot[fill=redfte,draw=white,line width=0.3pt,forget plot] coordinates {
    (-0.3, 0)
    (1, 4.294852939)
    (2, 4.434571943)
    (3, 2.928076268)
    (4, 6.247662089)
  };
  \node[font=\scriptsize,anchor=south] at (axis cs:-0.3,21.488396545) {18.7};
  \node[font=\scriptsize,anchor=south] at (axis cs:1,26.54204148) {23.7};
  \node[font=\scriptsize,anchor=south] at (axis cs:2,27.458549618) {24.6};
  \node[font=\scriptsize,anchor=south] at (axis cs:3,28.952707792) {26.1};
  \node[font=\scriptsize,anchor=south] at (axis cs:4,33.385019796) {30.6};
  \nextgroupplot[ylabel={Store latency (s) - \textbf{1GB}}]
  \addplot[fill=greenfte,draw=white,line width=0.3pt] coordinates {
    (-0.3, 36.378657684)
    (1, 39.571054338)
    (2, 45.715586575)
    (3, 45.452559667)
    (4, 33.543922481)
  };
  \addlegendentry[name=breakdown-label-walrus_store]{Walrus store}
  \addplot[fill=greyfte,draw=white,line width=0.3pt] coordinates {
    (-0.3, 0.457107495)
    (1, 0.388039247)
    (2, 0.358240351)
    (3, 0.341065705)
    (4, 0.451565836)
  };
  \addlegendentry[name=breakdown-label-walrus_client_setup]{Walrus setup}
  \addplot[fill=bluefte,draw=white,line width=0.3pt] coordinates {
    (-0.3, 0)
    (1, 2.097187017)
    (2, 2.085518241)
    (3, 1.945476295)
    (4, 1.439300933)
  };
  \addlegendentry[name=breakdown-label-share_gen]{Share generation}
  \addplot[fill=yellowfte,draw=white,line width=0.3pt] coordinates {
    (-0.3, 0)
    (1, 4.269341392)
    (2, 5.932270437)
    (3, 4.702427532)
    (4, 3.152981103)
  };
  \addlegendentry[name=breakdown-label-chain_create]{Sui create}
  \addplot[fill=purplefte,draw=white,line width=0.3pt] coordinates {
    (-0.3, 0)
    (1, 2.834931663)
    (2, 3.158707867)
    (3, 3.096494645)
    (4, 2.081093755)
  };
  \addlegendentry[name=breakdown-label-encrypt]{Encryption}
  \addplot[fill=redfte,draw=white,line width=0.3pt] coordinates {
    (-0.3, 0)
    (1, 4.103065761)
    (2, 5.399955956)
    (3, 2.945563309)
    (4, 4.437360857)
  };
  \addlegendentry[name=breakdown-label-chain_finalize]{Sui finalize}
  \node[font=\scriptsize,anchor=south] at (axis cs:-0.3,39.655028851) {36.8};
  \node[font=\scriptsize,anchor=south] at (axis cs:1,56.082882787) {53.3};
  \node[font=\scriptsize,anchor=south] at (axis cs:2,65.469542791) {62.7};
  \node[font=\scriptsize,anchor=south] at (axis cs:3,61.302850534) {58.5};
  \node[font=\scriptsize,anchor=south] at (axis cs:4,47.925488414) {45.1};
  \coordinate (breakdown-segment-walrus_store) at (axis cs:4,16.77196124);
  \coordinate (breakdown-segment-walrus_client_setup) at (axis cs:4,33.769705399);
  \coordinate (breakdown-segment-share_gen) at (axis cs:4,34.715138783);
  \coordinate (breakdown-segment-chain_create) at (axis cs:4,37.011279801);
  \coordinate (breakdown-segment-encrypt) at (axis cs:4,39.62831723);
  \coordinate (breakdown-segment-chain_finalize) at (axis cs:4,42.887544536);
  \end{groupplot}
  \draw[draw=black!35,line width=0.25pt] ([xshift=5.5pt]breakdown-segment-walrus_store) -- (breakdown-legend.west |- breakdown-label-walrus_store.west);
  \draw[draw=black!35,line width=0.25pt] ([xshift=5.5pt]breakdown-segment-walrus_client_setup) -- (breakdown-legend.west |- breakdown-label-walrus_client_setup.west);
  \draw[draw=black!35,line width=0.25pt] ([xshift=5.5pt]breakdown-segment-share_gen) -- (breakdown-legend.west |- breakdown-label-share_gen.west);
  \draw[draw=black!35,line width=0.25pt] ([xshift=5.5pt]breakdown-segment-chain_create) -- (breakdown-legend.west |- breakdown-label-chain_create.west);
  \draw[draw=black!35,line width=0.25pt] ([xshift=5.5pt]breakdown-segment-encrypt) -- (breakdown-legend.west |- breakdown-label-encrypt.west);
  \draw[draw=black!35,line width=0.25pt] ([xshift=5.5pt]breakdown-segment-chain_finalize) -- (breakdown-legend.west |- breakdown-label-chain_finalize.west);
\end{tikzpicture}
  \caption{Mean store phases over ten successful attempts; numbers give total seconds.
    Walrus storage includes registration, upload, and certification. Stack order groups
    phases rather than indicating their execution sequence.}
  \label{fig:macro-store-breakdown}
\end{figure}

\paragraph{Store and read latency results.}
\Cref{fig:macro-latency} shows that \sysname's additional store work does not translate directly
into slower reads. At 1\,GB with $(10,6)$, store latency is 58.48\,s versus Walrus's 36.84\,s, while
greedy reads take 10.03\,s versus 13.08\,s, a 23\% reduction; lazy reads take 12.74\,s. Greedy reads
also outperform Walrus at 100\,MB. We attribute this to better parallelization of the transfer: A
direct Walrus read fetches a single large blob from the set of storage nodes, whereas each \sysname
server reads a blob of size $|D|/k$ and returns it to the client, which utilizes the network more
evenly. Greedy reads additionally avoid waiting for the slowest servers. The advantage is not
universal: greedy reads with $(3,2)$ and $(5,3)$ remain slower than Walrus at 1\,GB, and stores are
slower at every tested size and configuration (\cref{app:macro-latency}). Larger $k$ also reduces
each share's payload: approximately $|D|/6$ with $(10,6)$, versus $|D|/2$ with $(3,2)$. We next
examine the extra store work in more detail.

\paragraph{Costs and $(n,k)$-layout tradeoff results.}
Walrus storage dominates store latency, accounting for 73--78\% of the 1\,GB total across all four
configurations (\cref{fig:macro-store-breakdown}). Sui creation and finalization (transactions
that create the on-chain content authorization material) come next at
7.59--11.33\,s; local share generation and encryption contribute 3.52--5.24\,s.
The choice of $(n,k)$ determines both the number of shares and their approximate payload size
$|D|/k$, and the aggregate expansion $n/k$. Configuration $(10,9)$ has the lowest 1\,GB store latency
and similar greedy-read latency to $(10,6)$, with expansion approximately $10/9$ rather than $10/6$. However, this lower redundancy tolerates one faulty share rather than four.

\section{Related Work}\label{sec:relatedwork}

\paragraph{Secret sharing.}
Shamir's Secret Sharing (SSS)~\cite{shamir1979share} is the prototypical secret sharing scheme. SSS
ensures that a secret can be shared across $n$ parties, such that $k$ of them are necessary to
recover the original secret, and does so with information-theoretic security. Its shares are as
large as the secret, and it provides no robustness. However, our robustness transform can be
applied to it as well.

\emph{Ramp secret sharing}~\cite{yamamoto1986secret,franklin1992communication} trades the
information-theoretic guarantees of SSS for smaller shares: $k$ shares reconstruct the secret, fewer
than $t < k$ shares leak nothing, and sets of size between $t$ and $k$ leak partial information.
Packed secret sharing~\cite{franklin1992communication} is an example. The share size of
approximately $\frac{1}{k-t}$ of the secret remains larger than the $\approx\frac{1}{k}$ achieved by
optimal CSS schemes.

\emph{AONT-RS}~\cite{resch2011aontrs} encodes data with an all-or-nothing transform
(AONT)~\cite{rivest1997all} followed by Reed-Solomon erasure coding, and Krawczyk's SSMS~\cite{krawczyk1993secret} encrypts the data with a random key, disperses the
ciphertext with an Information Dispersal Algorithm (IDA)~\cite{rabin1989efficient}, and shares the
key with SSS. Both are computational secret sharing schemes with share sizes of approximately
$\frac{1}{k}$ of the data, which Aureliano et
al.~\cite{aureliano2025optimalcomputationalsecretsharing} recently showed to be optimal for
computational secret sharing, and \sysname can use either as the underlying CSS scheme.

\paragraph{Robust secret sharing.}
Krawczyk~\cite{krawczyk1993secret} already provides robustness for SSMS with \emph{distributed
  fingerprints}~\cite{krawczyk1993distributed}: The dealer hashes each share and distributes the
hashes to the servers using an error-correcting code (ECC), such that a server cannot alter its
share without being detected. Bellare and Rogaway~\cite{bellare2007robust} formalize robust CSS,
prove Krawczyk's scheme (which they call HK1) secure in the random-oracle model, and replace the
hashes by commitments to obtain a variant (called HK2) that is secure under standard assumptions.
Chen et al.~\cite{chen2017revisiting} apply the same technique to AONT-RS, creating RAONT-RS. All of
these schemes inherit the drawback of the ECC, which may require a per-share authentication overhead
\emph{linear} in $n$.
To see this, consider the case where $n$ is even and $k = n/2+1$: The per-share authentication
overhead is $\approx 16n$B for all three ECC-based schemes, which is significant for large $n$,
whereas it is a fixed $96$B for \sysname. These overheads are fully analyzed and qualified
in~\cref{app:robustness-overhead}.

Krawczyk~\cite{krawczyk1993secret} also proposes dealer signatures on individual shares, which
require an authenticated dealer verification key, for example through a public-key infrastructure
(PKI). \sysname instead establishes an ephemeral verification key through agreement among the
shares, avoiding a pre-established key or PKI.

\emph{Verifiable Secret Sharing} (VSS)~\cite{rabin1989verifiable} is a stronger notion than robust secret
sharing, where parties can verify share consistency against a commitment to the secret and detect a
malicious dealer. However, VSS schemes typically require
interactive protocols and additional communication overhead, making them
less suitable for our setting, which assumes an honest dealer.

Bellare, Dai, and Rogaway~\cite{bellare2020reimagining} introduce \emph{adept secret sharing},
which adds privacy under non-uniform dealer randomness, authenticity, and error correction to
classical secret sharing. Their authenticity check operates on the reconstructed secret rather than
on individual shares, so recovery from adversarially corrupted shares may require inspecting up to
$2^n$ subsets, as the authors note; our transform verifies shares before decoding and avoids this
cost.

\paragraph{Key-Based Access Control on Public Storage.}
Proxy re-encryption~\cite{ateniese2006improved} and at\-tri\-bu\-te-based
encryption~\cite{goyal2006attribute} enforce access control by managing a decryption key. The data
resides encrypted in public storage, and the key is re-encrypted for, or encapsulated under a policy
for, authorized readers. Seal~\cite{mystenlabs2025seal} deserves a closer comparison, since it
targets the same setting as \sysname and shares its infrastructure pattern: A set of independent key
servers evaluates an on-chain policy for each request, and the protected data can live on Walrus.
The difference is what the servers protect. In Seal, the servers hold shares of a master secret from
which they derive, via threshold identity-based encryption, the decryption key for a given
ciphertext identity, and release it to authorized readers. The public storage therefore still holds
a ciphertext that any released key unlocks, so its redistribution bandwidth is $\approx |K|$
(\cref{sec:redistribution-bandwidth}). In \sysname, the data is secret-shared across the servers, so
a leaking reader must transmit $\approx |D|$.

Ideas similar to \sysname's \emph{redistribution bandwidth}
(\cref{sec:redistribution-bandwidth})---that large secrets are harder to exfiltrate than small
keys---are studied in Big-Key Cryptography~\cite{bellare2016bigkey}, which focuses on symmetric
encryption schemes with deliberately large keys to resist key exfiltration, as well as in the
Bounded Storage Model~\cite{cachin1997unconditional}, which establishes unconditional security for
key agreement under the assumption that the adversary's memory capacity is limited.

\section{Conclusion}

\sysname enables access control for decentralized storage, offering the security and high throughput
required to cater to the modern content-creation economy. \sysname increases the
\emph{redistribution bandwidth} required to leak protected data compared to existing approaches: An
adversary must exfiltrate the full content rather than a short key, turning a costless attack into
one whose cost scales with the size of the data itself. Its cryptographic core, \sharingname, adds
robustness to any computational secret-sharing scheme with constant-size authentication metadata per
share, without combinatorial trial decoding or a pre-authenticated dealer verification key. Our
end-to-end implementation on Walrus and Sui demonstrates the practicality of this approach, showing
encoding and reconstruction at tens of Gbps, and modest overhead on the geo-distributed store and
read latencies.

Because of these properties, we believe \sysname is uniquely positioned to unlock the potential of
decentralized storage systems, acting as an access control layer to publicly-stored data.

\clearpage

\bibliographystyle{splncs04}
\bibliography{ref}

\appendix
\crefalias{section}{appendix}
\crefalias{subsection}{subappendix}
\section{Pseudocode}\label{sec:pseudocode}

In this appendix, we show the pseudocode for the \sharingname algorithms. \Cref{alg:share} shows
the pseudocode for the $\RSSShare$ algorithm, and \cref{alg:reconstruct} shows the pseudocode for
the $\RSSReconstruct$ algorithm.

\begin{center}
  \begin{minipage}{0.75\textwidth}
  \begin{algorithm}[H]
    \caption{\sharingname $\RSSShare_\lambda$}
    \label{alg:share}
    \begin{algorithmic}[1]
      \Require Secret $s$
      \Ensure Shares $[\share{1}, \ldots, \share{n}]$

      \State $(\outsym{1}, \ldots, \outsym{n}) \leftarrow \CSSShare(s)$
      \Comment{Share with underlying CSS scheme}
      \State $(sk, pk) \leftarrow \text{KeyGen}(\lambda)$ \Comment{Generate ephemeral key pair}
      \For{$i = 1$ to $n$}  \Comment{Create the robust shares}
      \State $\sig{i} \leftarrow \sign{sk}{(i, \outsym{i})}$
      \State $\share{i} \leftarrow (pk, \sig{i}, i, \outsym{i})$
      \EndFor
      \State \textbf{return} $[\share{1}, \ldots, \share{n}]$
    \end{algorithmic}
  \end{algorithm}
\end{minipage}

  \vspace{1em}

  \begin{minipage}{0.75\textwidth}
  \begin{algorithm}[H]
    \caption{\sharingname $\RSSReconstruct$}
    \label{alg:reconstruct}
    \begin{algorithmic}[1]
        \Require Set of shares $S$
        \Ensure Original secret $s$ or $\bot$ (failure)

        \State $\text{pk\_counts} \leftarrow \{\}$ \Comment{Map to count public key occurrences}
        \For{$\share{i} \in S$} \Comment{Count votes for each public key}
        \State $(pk_i, \sig{i}, i, \outsym{i}) \leftarrow \share{i}$
        \State $\text{pk\_counts}[pk_i] \leftarrow \text{pk\_counts}[pk_i] + 1$
        \EndFor
        \State $pk \leftarrow \arg\max_{pk'} \text{pk\_counts}[pk']$ \Comment{Find most
        frequent public key}
        \If{$\text{pk\_counts}[pk] < k$}
        \State \textbf{return} $\bot$ \Comment{No agreement on public key}
        \EndIf
        \State $T \leftarrow \{\}$ \Comment{Verified indexed shares}
        \For{$\share{i} \in S$} \Comment{Verify shares and collect valid CSS shares}
        \State $(pk_i, \sig{i}, i, \outsym{i}) \leftarrow \share{i}$
        \If{$pk_i = pk$ \textbf{and} $\verify{pk}{(i, \outsym{i})}{\sig{i}} = \text{true}$}
        \State $T \leftarrow T \cup \{(i, \outsym{i})\}$
        \EndIf
        \EndFor
        \If{$|T| < k$}
        \State \textbf{return} $\bot$ \Comment{Not enough valid shares}
        \EndIf
        \State $s \leftarrow \CSSRecover(T)$ \Comment{Recover with underlying CSS scheme}
        \State \textbf{return} $s$
    \end{algorithmic}
  \end{algorithm}
\end{minipage}

\end{center}

\section{Security Games and Proofs}\label{sec:proofs}

In this appendix, we prove that the robustness transform turns any secure CSS scheme
(\cref{def:css}) into a robust secret sharing scheme. Our analysis assumes an honest dealer (see
\cref{sec:system-model}) and consists of two theorems that hold for \emph{any} underlying CSS
scheme: The transform preserves the secrecy of the underlying scheme, i.e., it introduces no
additional leakage (\cref{thm:secrecy-preservation}), and it provides robust correctness, by
reduction to the unforgeability of the signature scheme (\cref{thm:robust-correctness}).

\subsection{Definitions}\label{sec:definitions}

We first restate the definitions of computational and robust computational secret sharing from
\cref{sec:css,sec:robustness-transform} formally.

\begin{definition}[Computational Secret Sharing]\label{def:css}
  A $(k, n)$ \emph{computational secret sharing (CSS) scheme} with security parameter $\lambda$ is
  a tuple of efficient algorithms $(\CSSShare, \CSSRecover)$ with the following properties:
  \begin{itemize}
    \item $\CSSShare_\lambda(s) \rightarrow (\outsym{1}, \ldots, \outsym{n})$ takes a secret $s$
            and outputs $n$ indexed shares.
          \item $\CSSRecover_\lambda(T) \rightarrow s'$ takes a set of indexed shares $T \subseteq
              \{(i, \outsym{i})\}$ and outputs either a secret $s'$ or $\bot$ (failure).
          \end{itemize}
  The scheme satisfies:
  \begin{itemize}
    \item \textbf{Completeness:} For every secret $s$ and \emph{every} subset $T$ containing at
            least $k$ of the indexed shares output by $\CSSShare_\lambda(s)$,
            $\CSSRecover_\lambda(T) = s$, except with negligible probability in $\lambda$.
          \item \textbf{Secrecy:} Any set of fewer than $k$ shares reveals no information about the
          secret $s$ to a computationally bounded adversary. This is captured formally by the CSS
            indistinguishability game in \cref{def:css-ind}.
          \end{itemize}
\end{definition}

\begin{definition}[Robust Computational Secret Sharing]\label{def:rss}
  A $(k, n)$ \emph{robust computational secret sharing (RSS)} scheme
    $(\RSSShare, \RSSReconstruct)$ is a CSS scheme (\cref{def:css}) that additionally
    provides \emph{robust correctness}: If at most $f = \min(k-1, n-k)$ shares are
    corrupted and at least $k$ original shares remain, $\RSSReconstruct$ outputs the original
    secret, except with negligible probability. This is formalized by the secret-malleability game
    in \cref{def:rss-malleability}.
  \end{definition}

\subsection{Security Games}\label{sec:rss-games}

The CSS indistinguishability game captures the secrecy of a CSS scheme; it is the property the
robustness transform must preserve.
\begin{definition}[CSS Indistinguishability Game]\label{def:css-ind}
  Given a $(k, n)$ computational secret sharing scheme $(\CSSShare_\lambda, \CSSRecover_\lambda)$,
  the CSS indistinguishability (CSS-IND) game between a probabilistic polynomial-time adversary
  $\mathcal{A}$ and a challenger proceeds as follows:
  \begin{enumerate}
    \item The adversary $\mathcal{A}$ selects two secrets $s_0$ and $s_1$ of the same length and submits them
          to the challenger.
    \item The challenger samples a random bit $b \leftarrow \{0,1\}$ and computes $(\outsym{1}, \ldots,
            \outsym{n}) \gets \CSSShare_\lambda(s_b)$.
    \item The adversary is allowed to adaptively query the challenger for any subset of at most $k-1$ shares;
          for each query $i \in \{1, \ldots, n\}$, the challenger returns $\outsym{i}$.
    \item The adversary outputs a guess $b' \in \{0,1\}$.
    \item The adversary wins if $b' = b$. In this case, the game outputs $1$, otherwise it outputs $0$.
  \end{enumerate}
  The CSS scheme provides indistinguishability if for all PPT adversaries $\mathcal{A}$,
  the advantage
  \[
    \mathrm{Adv}_{\mathcal{A}} = \left| \Pr[\mathrm{CSS\text{-}IND}(\mathcal{A}) = 1] -
    \frac{1}{2} \right|
  \]
  is negligible in the security parameter $\lambda$.
\end{definition}

The secret-malleability game captures robust correctness: An adversary who corrupts up to $f$
shares can neither change the reconstructed secret nor prevent reconstruction while $k$ original
shares remain.

\begin{definition}[Secret-Malleability Game]\label{def:rss-malleability}
  Given a $(k, n)$ robust secret sharing scheme $(\RSSShare_\lambda, \RSSReconstruct_\lambda)$,
  the secret-malleability game between a probabilistic polynomial-time adversary
  $\mathcal{A}$ and a challenger proceeds as follows:
  \begin{enumerate}
    \item The adversary chooses an arbitrary secret $s$ and requests the challenger to generate a set $S$ of
          $n$ shares using $\RSSShare_\lambda(s)$.
    \item The challenger generates the shares and provides them to the adversary.
    \item The adversary produces a modified set of shares $S'$ by arbitrarily corrupting and/or dropping
          shares, subject to the constraint that at most $f$ shares are corrupted.
    \item The challenger runs $\RSSReconstruct_\lambda$ on the modified shares $S'$ to obtain a reconstructed
          secret $s'$ (which may be $\bot$).
    \item The adversary wins if either:
          \begin{enumerate}
            \item $s' \neq s$ and $s' \neq \bot$ (successful corruption), or
            \item $s' = \bot$ when $S'$ contains at least $k$ original shares (incorrect rejection)
          \end{enumerate}
  \end{enumerate}
  The scheme provides robust correctness if for all PPT adversaries $\mathcal{A}$, the
  probability that the adversary wins the game is negligible in the security
  parameter $\lambda$.
\end{definition}

\subsection{Secrecy Preservation}%
\label{sec:secrecy-preservation}

We first show that the robustness transform introduces no leakage beyond that of the underlying
scheme. The only additional values contained in the shares are the ephemeral public key and the
signatures, the former of which is computed independently of the secret, and the latter only
depends on the single share that it accompanies.%
\begin{theorem}[Secrecy Preservation]\label{thm:secrecy-preservation}
  Let $\css = (\CSSShare, \CSSRecover)$ be a $(k, n)$ CSS scheme satisfying CSS
  indistinguishability (\cref{def:css-ind}). Then the scheme obtained by applying the robustness
  transform to $\css$ also satisfies CSS indistinguishability.
\end{theorem}

\begin{proof}
  Towards a contradiction, assume that there is a PPT adversary $\mathcal{A}$ that wins the CSS-IND
  game against the transformed scheme with non-negligible advantage. We construct an adversary
  $\mathcal{B}$ against the underlying scheme $\css$ as follows: $\mathcal{B}$ forwards the two
  secrets $s_0$ and $s_1$ chosen by $\mathcal{A}$ to its own challenger and generates an ephemeral
  key pair $(sk, pk)$. When $\mathcal{A}$ queries the share at index $i$, $\mathcal{B}$ queries its
  own challenger for the CSS share $\outsym{i}$ and returns the robust share $\share{i} = (pk,
    \sign{sk}{(i, \outsym{i})}, i, \outsym{i})$ to $\mathcal{A}$.%

  To do so, it makes at most the same $k-1$ queries as $\mathcal{A}$. Finally, $\mathcal{B}$ outputs
  the bit output by $\mathcal{A}$.

  Since the ephemeral key pair is sampled independently of the secret, exactly as in the real scheme,
  $\mathcal{B}$ perfectly simulates the CSS-IND game for the transformed scheme, and the advantage of
  $\mathcal{B}$ equals the advantage of $\mathcal{A}$. This contradicts the assumption that $\css$
  satisfies CSS indistinguishability, and thus no such adversary $\mathcal{A}$ can exist. \qed
\end{proof}

\subsection{Robust Correctness}%
\label{sec:robust-correctness}

We now show that the robustness transform provides robust correctness for any underlying CSS
scheme.%
\begin{theorem}[Robust Correctness]\label{thm:robust-correctness}
  Let $\css$ be a $(k, n)$ CSS scheme (\cref{def:css}) and let the signature scheme used by the
  robustness transform be EUF-CMA-secure. Then, no computationally bounded adversary can win the
  secret-malleability game (\cref{def:rss-malleability}%
  ) against the scheme obtained by applying
  the robustness transform to $\css$ with non-negligible probability.
\end{theorem}

\begin{proof}
  Towards a contradiction, assume that an adversary $\mathcal{A}$ exists that wins the
  secret-malleability game with non-negligible probability. We will use this adversary to construct
  an adversary $\mathcal{A}'$ that breaks the unforgeability of the underlying signature scheme with
  non-negligible probability.

  Adversary $\mathcal{A}'$ works as follows:
  \begin{enumerate}
    \item $\mathcal{A}'$ receives a public key $pk$ from the challenger of the signature
          forgery game.
    \item $\mathcal{A}'$ simulates the secret-malleability game for $\mathcal{A}$:
          \begin{itemize}
            \item When $\mathcal{A}$ requests shares for a secret $s$, $\mathcal{A}'$ generates the unsigned shares
                  $(i, \outsym{i})$ using $\CSSShare(s)$, as in the protocol description in
                  \cref{sec:robustness-transform}.
            \item $\mathcal{A}'$ uses the oracle of the signature forgery game to request
                  signatures on the unsigned shares to obtain robust shares that use the received
                  public key.
            \item $\mathcal{A}'$ provides $\mathcal{A}$ with the robust shares.
          \end{itemize}
    \item When $\mathcal{A}$ wins the game, there are two possibilities:

          \emph{Case 1 (Successful Corruption):} If $\mathcal{A}$ produces $k$ valid
          shares that reconstruct to $s' \neq s$, then all $k$ of them are verifiable with
          the same public key. Since $\mathcal{A}$ can only replace fewer
          than $k$ shares (since $f = \min(k-1, n-k)$), at least one of the
          $k$ shares used for reconstruction must be an original share signed with the
          original key from the signature oracle and thus, the key $pk'$ used for all $k$
          shares must be the key $pk$ received from the challenger of the signature
          forgery game. Since the $k$ shares reconstruct to a different secret, at least
          one of them is not an original share, and thus is a message/signature pair that
          was never requested from the signing oracle.

          \emph{Case 2 (Incorrect Rejection):} Since at least $k$ shares are original,
          $\mathcal{A}$ must have provided at least $k$ shares that are not discarded,
          i.e. they agree on the same public key $pk$ and have a valid signature under
          that key, otherwise the reconstruction succeeds trivially. However,
          reconstruction from these shares fails, even though, by the completeness of the
          underlying CSS scheme (\cref{def:css}), $\CSSRecover$ outputs the original
          secret on any set of at least $k$ authentic indexed shares (except with
          negligible probability). Therefore, one of the verified shares differs from the
          original. This share is again a message/signature pair that was never requested
          from the signing oracle.

    \item $\mathcal{A}'$ outputs the message/signature pair obtained in the previous step.
  \end{enumerate}

  The adversary $\mathcal{A}'$ can therefore always extract a forged signature if $\mathcal{A}$ wins
  the secret-malleability game, and wins the signature forgery game with the same non-negligible
  probability with which $\mathcal{A}$ wins the secret-malleability game.

  This contradicts our assumption that the underlying signature scheme is unforgeable, and thus no
  such adversary $\mathcal{A}$ exists. \qed
\end{proof}

\subsection{Instantiations}\label{sec:instantiations}

Two established CSS schemes for large data satisfy \cref{def:css} and produce shares of size
$\approx \frac{|s|}{k}$; we use both in our evaluation (\cref{sec:evaluation}). In
\emph{AONT-RS}~\cite{resch2011aontrs}, $\CSSShare$ applies an all-or-nothing transform
(AONT)~\cite{rivest1997all} to the secret, which is padded to at least $k$ blocks such that every
erasure-code input share contains at least one AONT block, and encodes the result with an $(n, k)$
maximum distance separable (MDS) erasure code; $\CSSRecover$ erasure-decodes the shares and inverts
the AONT. In \emph{SSMS}~\cite{krawczyk1993secret}, $\CSSShare$ encrypts the secret with a freshly
sampled symmetric key $K$, disperses the ciphertext with an $(n, k)$ MDS code (an information
dispersal algorithm~\cite{rabin1989efficient}), and shares $K$ with Shamir's secret
sharing~\cite{shamir1979share}; share $\outsym{i}$ consists of the $i$-th ciphertext share and the
$i$-th key share, and $\CSSRecover$ interpolates $K$ from any $k$ key shares, erasure-decodes the
ciphertext, and decrypts it. Completeness (\cref{def:css}) of both schemes follows directly, since
MDS decoding from any $k$ shares, AONT decoding, Shamir interpolation, and decryption are
deterministic.

Since \cref{thm:secrecy-preservation,thm:robust-correctness} hold for any secure CSS scheme, the
secrecy of the two schemes we use for our evaluation, SSMS~\cite{krawczyk1993secret} and
AONT-RS~\cite{resch2011aontrs}, is inherited from the underlying schemes rather than proven from
scratch: For SSMS, secrecy was shown by Krawczyk~\cite{krawczyk1993secret} and formally analyzed by
Bellare and Rogaway~\cite{bellare2007robust}; for AONT-RS, we rely on the analysis by Chen et
al.~\cite{chen2017revisiting}, whose information-leakage attack for small secrets does not apply
here because of the padding to at least $k$ blocks.

\section{Additional Evaluation Details}\label{app:evaluation}

\subsection{Microbenchmark Methodology}\label{app:micro-methodology}
\paragraph{Hardware.}
All microbenchmarks, including the throughput and cost-breakdown experiments, run on an Apple M2
Max.

\paragraph{Comparison implementations.}
HK1, the HK2-style variant, and RAONT-RS are implemented in Rust with the same symmetric
encryption, hashing, and dispersal primitives as \sharingname. They authenticate shares through
fingerprints or commitments encoded with an error-correcting code, rather than signatures. HK1 uses
deterministic BLAKE3 fingerprints. The implementation labeled HK2 uses BLAKE3 with a 16-byte random
opening in place of the statistically hiding commitment required by the formal HK2 construction.
Its performance therefore describes this variant, not an implementation of that construction's full
guarantees.

\paragraph{Cases and inputs.}
Both experiments cover $(n,k)=(3,2),(5,3),(10,6),(10,9)$, 1 and 12 Rayon workers, and 1, 10, 100,
and 1,000\,MB, with $1\,\mathrm{MB}=10^6$ bytes. The three operations are encoding, first-$k$
reconstruction, and last-$k$ reconstruction. The throughput experiment contains 672 cases across
seven schemes (two non-robust schemes are omitted from the plots); the separate cost-breakdown
experiment contains 480 across the five robust schemes. Uncorrupted shares are valid encoder outputs
without adversarial modification. For the SSMS- and AONT-RS-based schemes, the first $k$ shares
contain systematic dispersal shards, avoiding erasure reconstruction of the bulk data. In these
layouts, the last $k$ omit $n-k$ systematic shards. Both paths still authenticate shares in the
robust schemes and decrypt the data; SSMS also reconstructs its key. These are not Byzantine
best/worst cases.

\paragraph{Throughput timing.}
Criterion uses 20 samples per case with flat sampling. Each timed invocation receives a copy of a
fixed, correctness-checked input or share set. Algorithm-internal allocations, randomness, and
authentication are included; input generation and copying, correctness checks, output destruction,
and thread-pool setup and dispatch are excluded. Timing covers only the sharing layer, before
serialization, per-server encryption of shares, networking, and access control. Throughput counts
original plaintext bits, not encoded size, in decimal Gbps ($10^9$ bits/s). Confidence intervals
describe within-run sampling uncertainty, not cross-run variation.

\paragraph{Cost breakdown and controls.}
The independent cost-breakdown experiment collects 20 paired samples per case. Each invocation
records its total duration and non-overlapping base-operation spans. The SSMS base includes
encryption or decryption, ciphertext dispersal or recovery, and key sharing or reconstruction. The
AONT-RS base includes the AONT, its ciphertext hash, and dispersal or recovery. Robustness overhead
is total time minus base time within the same invocation; it includes signature generation or
verification (commitment encoding or error-correction decoding for the baselines), share
verification, and assembling the robust shares. Bootstrap intervals use 5,000 resamples of total
time; they are not sums of phase intervals. We do not subtract or combine timings from independently
benchmarked schemes or the throughput run.

We interleave calls through the uninstrumented API to check measurement perturbation. Across all
480 cases, the absolute difference between profiled and control mean totals, relative to the
control mean, has median 0.24\% and maximum 5.39\%. These differences can include timing noise and
do not measure robustness overhead. The paired experiment's one encoding case where our SSMS total
exceeds our AONT total is $(3,2)$, 12 workers, 100\,MB. Across all SSMS cases, robustness overhead
accounts for 38.8--84.0\% of first-$k$ time and 11.8--51.2\% of last-$k$ time. Constant-size
authentication metadata does not imply constant authentication time: hashing still processes share
payloads.

\paragraph{Additional throughput comparisons.}
Across the full matrix, our SSMS instantiation exceeds RAONT-RS throughput in all 32 encoding cases
and 28 of 32 cases for each reconstruction operation. First-$k$ exceeds last-$k$ throughput for
both of our instantiations in all 32 matched cases. In the $(10,6)$, 12-worker, 1\,GB selection,
SSMS first-$k$ and last-$k$ rates are 77.6 and 33.8\,Gbps, compared with 41.7 and 25.7\,Gbps for
AONT-RS; SSMS is within 6\% of HK1 and HK2 for all three operations. At 1\,MB in that selection,
first-$k$ instead reaches 10.4\,Gbps with SSMS versus 13.7\,Gbps with AONT-RS. These examples do
not replace the full-matrix comparisons in the main text.

\subsection{Microbenchmarks Across All Configurations}\label{app:micro-matrix}

\Cref{fig:micro-all-3-2,fig:micro-all-5-3,fig:micro-all-10-6,fig:micro-all-10-9} extend the
main-text plots to all four $(n,k)$ configurations with 12 workers. Each figure shows throughput
across all four payload sizes above the cost breakdown at 1\,GB. Columns show encoding,
first-$k$, and last-$k$ reconstruction, using uncorrupted shares. The five robust schemes and
axis limits are consistent across configurations. The two \sysname bars in each breakdown use
SSMS and AONT-RS, respectively; their upper segments show robustness overhead.

Throughput whiskers show Criterion 95\% confidence intervals; breakdown whiskers show 95\%
bootstrap intervals for mean total time from the independent cost-breakdown experiment. The
full-matrix comparisons in the main text also include single-worker measurements and all payload
sizes in the cost-breakdown experiment, beyond the selections plotted here. Counts of higher
estimates describe the observed ordering, not statistical significance.

\begingroup
\pgfplotsset{legend to name/.code={\pgfkeyssetvalue{/pgfplots/legend to name}{appendix-#1}}}
\let\evaluationlegendfromname\pgfplotslegendfromname
\renewcommand{\pgfplotslegendfromname}[1]{\evaluationlegendfromname{appendix-#1}}

\begin{figure}[!htbp]
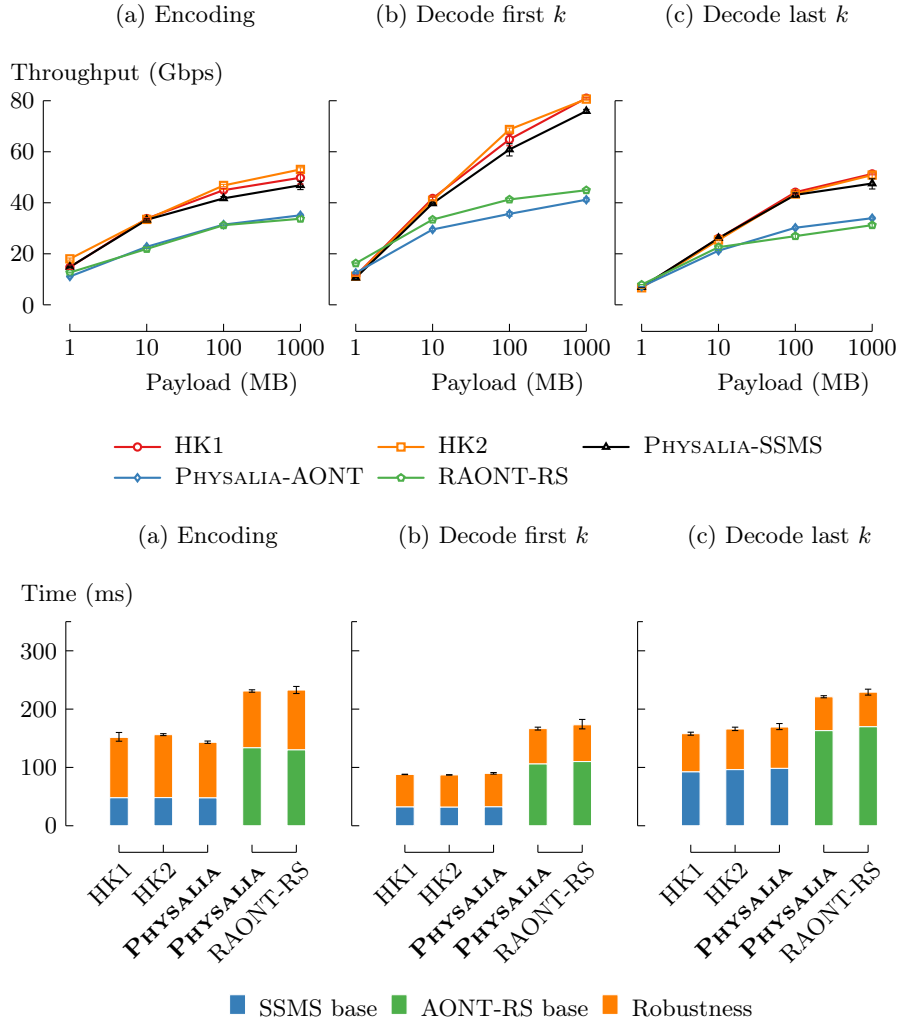

  \centering
  \input{figures/micro/throughput_3_2}
  \par\medskip
  \input{figures/micro/breakdown_3_2}
  \caption{Throughput (top) and 1\,GB cost breakdown (bottom) for $(n,k)=(3,2)$ with 12 workers.}
  \label{fig:micro-all-3-2}
\end{figure}

\begin{figure}[!htbp]
  \centering
  \input{figures/micro/throughput_5_3}
  \par\medskip
  \input{figures/micro/breakdown_5_3}
  \caption{Throughput (top) and 1\,GB cost breakdown (bottom) for $(n,k)=(5,3)$ with 12 workers.}
  \label{fig:micro-all-5-3}
\end{figure}

\begin{figure}[!htbp]
  \centering
  \input{figures/micro/throughput_10_6}
  \par\medskip
  \input{figures/micro/breakdown_10_6}
  \caption{Throughput (top) and 1\,GB cost breakdown (bottom) for $(n,k)=(10,6)$ with 12 workers.}
  \label{fig:micro-all-10-6}
\end{figure}

\begin{figure}[!htbp]
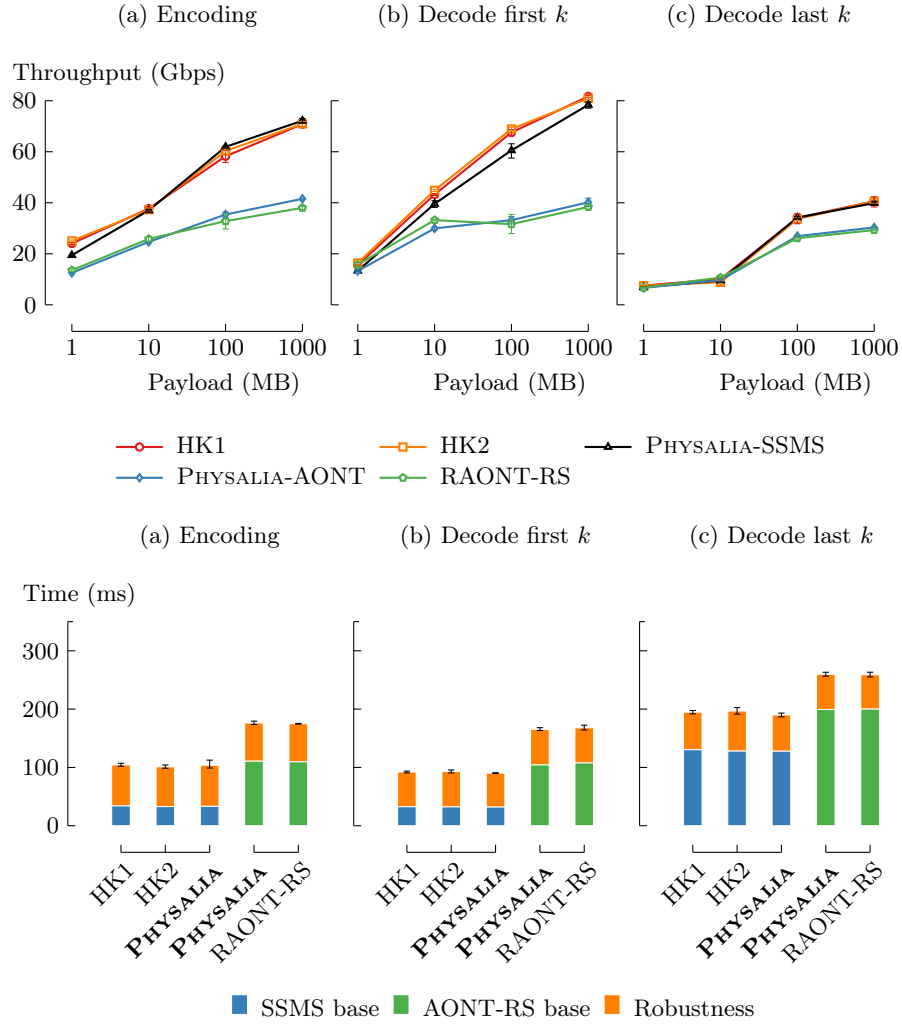

  \centering
  \input{figures/micro/throughput_10_9}
  \par\medskip
  \input{figures/micro/breakdown_10_9}
  \caption{Throughput (top) and 1\,GB cost breakdown (bottom) for $(n,k)=(10,9)$ with 12 workers.}
  \label{fig:micro-all-10-9}
\end{figure}

\FloatBarrier
\endgroup

\subsection{End-to-End Methodology}\label{app:macro-methodology}
\paragraph{Deployment and baseline.}
The client and each \sysname server run on AWS \texttt{m7i.2xlarge} instances. The client is in
\texttt{eu-west-1}; servers are assigned round-robin to \texttt{us-west-1}, \texttt{us-east-1}, and
\texttt{eu-west-1}. Each layout uses the first $n$ servers, giving the ten-server layout four
servers in US West, three in US East, and three in Europe. Each operation starts a fresh client
process; servers remain running. The prototype maintains access policies on Sui testnet and stores
encrypted shares on Walrus testnet. The baseline stores plaintext as a single Walrus blob through
the same SDK, without \sysname's access-control layer. Both contact storage nodes directly, without
an intermediate publisher or aggregator.

\paragraph{Walrus configuration.}\label{app:walrus-config}
Across all layouts, we use the SDK's balanced profile for \sysname (up to 1,000 concurrent writes
and 512\,MB in flight) and its aggressive profile for the baseline (2,000 writes and 1.25\,GB),
without overriding concurrency settings. For each system, the chosen profile performed better than
the other one. Both request deletable storage for one epoch and disable reuse of existing blob and
storage resources. The comparison reflects both system design and resource allocation, not an
isolated SDK-profile effect.

\paragraph{Procedure.}
Each store receives a newly generated random payload. Each size/layout has ten stores, ten greedy
reads, and ten lazy reads. Each strategy reuses one stored object, with distinct objects for greedy
and lazy reads. The baseline has ten stores and ten reads per size, reusing one blob for reads. All
observations, including first reads, are retained: these are repeated accesses to fixed objects,
not independent trials on newly stored content. Greedy requests all $n$ shares concurrently and
cancels remaining requests once reconstruction succeeds. Lazy initially contacts the first $k$
configured servers and requests more if a request or reconstruction fails.

The 560 operations execute sequentially, grouped by size and layout, with baseline operations
interleaved and no randomized ordering. The controller permits up to five attempts per operation
with a 10\,s delay. All operations completed; one store required a second attempt after a Sui RPC
timeout. The run uses software revision \texttt{b75369b6}.

\paragraph{Timing boundaries.}
We report client-side wall-clock time for successful attempts. A \sysname store includes share
generation and per-server encryption, Walrus client initialization and storage, and creation and
finalization of the Sui content object. A baseline store includes Walrus client initialization and
storage. Reads include client initialization and complete-payload retrieval, plus server-side
authorization and reconstruction for \sysname. Payload generation, process launch, and result
logging are excluded, as are failed attempts and controller retry delays. SDK retries within
successful attempts remain included. Plots show arithmetic means and individual observations. These
sequential testnet measurements do not evaluate concurrent throughput or injected failures.

\paragraph{Additional observations.}
For $(10,6)$ at 1\,GB, the response completing reconstruction came from US East in all ten greedy
reads and US West in all ten lazy reads. This observation depends on placement and request order;
it does not separate network effects from read strategy. At 100\,MB, mean greedy latency was
3.17\,s versus Walrus's 3.33\,s. Across layouts at 1\,GB, Sui creation and finalization take
7.59--11.33\,s, while share generation and encryption take 3.52--5.24\,s. In the store breakdown,
Walrus setup denotes SDK client initialization; Walrus storage includes registration, upload, and
certification.

\subsection{Latency Across All Layouts}\label{app:macro-latency}

\Cref{fig:macro-read-latency,fig:macro-store-latency} report all four layouts from the same run and
measurement procedure as \cref{sec:macro-experiments}. The main-text comparison in
\cref{fig:macro-latency} selects $(n,k)=(10,6)$ from these measurements without changing the
samples.

\begin{figure}[!htbp]
  \centering
  \input{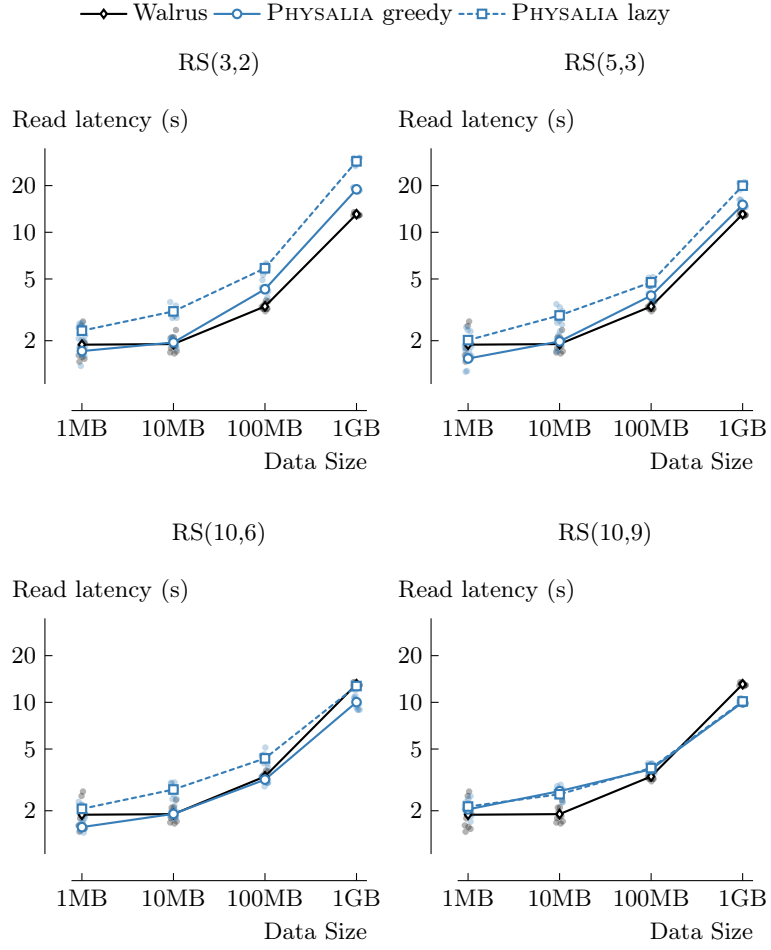}
  \caption{End-to-end read latency by layout. Lines show means of ten successful attempts; faint
    dots show individual observations with horizontal jitter. First reads are included. The same
    Walrus observations appear in each panel. Both axes are logarithmic.}
  \label{fig:macro-read-latency}
\end{figure}

\begin{figure}[!htbp]
  \centering
  \begin{tikzpicture}
  \begin{axis}[
    height=5.5cm, width=0.72\linewidth,
    ylabel={Store latency (s)},
    xlabel={Data Size},
    xmode=log, log basis x=10,
    xmin=0.769230769, xmax=1300,
    xtick={1,10,100,1000},
    xticklabels={1MB,10MB,100MB,1GB},
    ymin=0, ymax=98.250296182,
    legend style={at={(1.05,1)},anchor=north west,legend columns=1,
      font=\footnotesize,legend cell align=left,draw=none},
  ]
  \addplot[only marks,mark=*,mark size=0.8pt,color=black,opacity=0.3,mark options={solid,fill=black,draw=none},forget plot] coordinates {
    (0.922571427, 17.024649572)
    (0.939242591, 14.584686601)
    (0.956215008, 14.566622477)
    (0.973494121, 13.339752665)
    (0.991085475, 17.823388017)
    (1.008994709, 18.56626489)
    (1.027227569, 21.734326461)
    (1.045789903, 19.642566429)
    (1.064687664, 14.828885022)
    (1.083926914, 34.580197231)
    (9.225714272, 19.709636411)
    (9.392425907, 16.793931479)
    (9.562150076, 13.860573208)
    (9.734941215, 17.854948615)
    (9.910854745, 12.525759518)
    (10.08994709, 13.411086131)
    (10.272275692, 20.916330395)
    (10.45789903, 19.043895301)
    (10.646876642, 13.76552039)
    (10.83926914, 17.669564846)
    (92.257142715, 29.246441513)
    (93.92425907, 23.343560438)
    (95.621500755, 17.885472771)
    (97.349412146, 15.783537764)
    (99.108547453, 20.603647981)
    (100.899470903, 14.776187398)
    (102.722756918, 17.705736581)
    (104.5789903, 17.055078509)
    (106.468766419, 18.052440782)
    (108.392691402, 17.424016872)
    (922.571427155, 32.999344844)
    (939.2425907, 35.288406432)
    (956.215007553, 38.412333303)
    (973.494121458, 41.412993255)
    (991.085474531, 33.898662765)
    (1008.994709032, 42.790341412)
    (1027.227569183, 33.389064531)
    (1045.789903004, 42.63207302)
    (1064.68766419, 33.221692068)
    (1083.926914021, 34.312750798)
  };
  \addplot[only marks,mark=*,mark size=0.8pt,color=bluefte,opacity=0.3,mark options={solid,fill=bluefte,draw=none},forget plot] coordinates {
    (0.922571427, 23.286174525)
    (0.939242591, 20.592379493)
    (0.956215008, 28.667119361)
    (0.973494121, 23.612610477)
    (0.991085475, 17.739298887)
    (1.008994709, 21.387153597)
    (1.027227569, 30.292512228)
    (1.045789903, 24.990917205)
    (1.064687664, 20.691958998)
    (1.083926914, 25.967663944)
    (9.225714272, 19.529162875)
    (9.392425907, 30.937958191)
    (9.562150076, 25.068457764)
    (9.734941215, 37.030636742)
    (9.910854745, 22.596192367)
    (10.08994709, 31.968557428)
    (10.272275692, 24.275066097)
    (10.45789903, 37.495626597)
    (10.646876642, 20.920178336)
    (10.83926914, 23.013361031)
    (92.257142715, 24.011983547)
    (93.92425907, 25.770292819)
    (95.621500755, 26.804827771)
    (97.349412146, 27.05415339)
    (99.108547453, 29.526451717)
    (100.899470903, 28.206556559)
    (102.722756918, 29.242525086)
    (104.5789903, 28.022591357)
    (106.468766419, 26.583365892)
    (108.392691402, 23.138238093)
    (922.571427155, 42.246940954)
    (939.2425907, 47.740781943)
    (956.215007553, 62.501585326)
    (973.494121458, 56.971087699)
    (991.085474531, 50.058619502)
    (1008.994709032, 70.028625621)
    (1027.227569183, 50.579362308)
    (1045.789903004, 50.100003466)
    (1064.68766419, 48.244645765)
    (1083.926914021, 54.164549208)
  };
  \addplot[only marks,mark=*,mark size=0.8pt,color=bluefte,opacity=0.3,mark options={solid,fill=bluefte,draw=none},forget plot] coordinates {
    (0.922571427, 27.075758157)
    (0.939242591, 23.481143533)
    (0.956215008, 21.784875601)
    (0.973494121, 26.146468285)
    (0.991085475, 36.020762742)
    (1.008994709, 15.405284494)
    (1.027227569, 18.398998308)
    (1.045789903, 22.979331393)
    (1.064687664, 22.774087302)
    (1.083926914, 32.326160286)
    (9.225714272, 33.259903509)
    (9.392425907, 17.879339408)
    (9.562150076, 19.938041222)
    (9.734941215, 21.576238525)
    (9.910854745, 23.122140664)
    (10.08994709, 25.260618763)
    (10.272275692, 24.265037764)
    (10.45789903, 35.479791105)
    (10.646876642, 19.046101078)
    (10.83926914, 21.552153366)
    (92.257142715, 20.853951014)
    (93.92425907, 25.750895676)
    (95.621500755, 24.238730418)
    (97.349412146, 26.821132825)
    (99.108547453, 31.894412737)
    (100.899470903, 19.638754076)
    (102.722756918, 27.824709387)
    (104.5789903, 27.190492295)
    (106.468766419, 25.836235943)
    (108.392691402, 20.125816249)
    (922.571427155, 65.70179137)
    (939.2425907, 66.358165256)
    (956.215007553, 46.606517279)
    (973.494121458, 61.76900005)
    (991.085474531, 74.040627249)
    (1008.994709032, 69.012153647)
    (1027.227569183, 62.966138122)
    (1045.789903004, 63.159367082)
    (1064.68766419, 59.951003803)
    (1083.926914021, 56.938037966)
  };
  \addplot[only marks,mark=*,mark size=0.8pt,color=bluefte,opacity=0.3,mark options={solid,fill=bluefte,draw=none},forget plot] coordinates {
    (0.922571427, 29.919555447)
    (0.939242591, 35.120765528)
    (0.956215008, 21.40705891)
    (0.973494121, 20.775813902)
    (0.991085475, 22.628263177)
    (1.008994709, 28.642121612)
    (1.027227569, 19.739525766)
    (1.045789903, 22.50193444)
    (1.064687664, 28.680321645)
    (1.083926914, 31.919091413)
    (9.225714272, 20.735619581)
    (9.392425907, 30.705729968)
    (9.562150076, 21.91548686)
    (9.734941215, 27.036956159)
    (9.910854745, 23.123405053)
    (10.08994709, 19.726884239)
    (10.272275692, 17.308712362)
    (10.45789903, 20.93836984)
    (10.646876642, 23.7348435)
    (10.83926914, 19.988383208)
    (92.257142715, 20.692925831)
    (93.92425907, 18.405411265)
    (95.621500755, 25.991053531)
    (97.349412146, 20.670581756)
    (99.108547453, 20.454332461)
    (100.899470903, 22.454184507)
    (102.722756918, 28.91348746)
    (104.5789903, 22.343087808)
    (106.468766419, 20.075316291)
    (108.392691402, 31.30856892)
    (922.571427155, 62.585122202)
    (939.2425907, 48.506681781)
    (956.215007553, 49.182245429)
    (973.494121458, 85.435040158)
    (991.085474531, 66.968061563)
    (1008.994709032, 52.91084609)
    (1027.227569183, 61.531482102)
    (1045.789903004, 49.612418518)
    (1064.68766419, 56.123725931)
    (1083.926914021, 51.980255488)
  };
  \addplot[only marks,mark=*,mark size=0.8pt,color=bluefte,opacity=0.3,mark options={solid,fill=bluefte,draw=none},forget plot] coordinates {
    (0.922571427, 25.335483936)
    (0.939242591, 44.728607737)
    (0.956215008, 54.096071271)
    (0.973494121, 23.906181982)
    (0.991085475, 32.530724991)
    (1.008994709, 26.43565713)
    (1.027227569, 25.670270252)
    (1.045789903, 22.902834283)
    (1.064687664, 27.952494325)
    (1.083926914, 22.099245966)
    (9.225714272, 23.110649516)
    (9.392425907, 23.717597056)
    (9.562150076, 19.342678063)
    (9.734941215, 21.049428837)
    (9.910854745, 32.12107477)
    (10.08994709, 27.337304329)
    (10.272275692, 22.969468773)
    (10.45789903, 29.063204219)
    (10.646876642, 31.052697687)
    (10.83926914, 28.836679484)
    (92.257142715, 20.563154178)
    (93.92425907, 34.245049956)
    (95.621500755, 19.277244394)
    (97.349412146, 18.499445698)
    (99.108547453, 18.856385812)
    (100.899470903, 22.073215763)
    (102.722756918, 21.905692049)
    (104.5789903, 19.853138761)
    (106.468766419, 21.848217009)
    (108.392691402, 18.257528906)
    (922.571427155, 49.528243423)
    (939.2425907, 43.345165513)
    (956.215007553, 43.65938492)
    (973.494121458, 46.981000313)
    (991.085474531, 39.062466689)
    (1008.994709032, 48.542518126)
    (1027.227569183, 39.817983981)
    (1045.789903004, 50.203614141)
    (1064.68766419, 51.567510237)
    (1083.926914021, 38.354370717)
  };
  \addplot[black,solid,mark=diamond*,mark size=2pt] coordinates {
    (1, 18.669133936)
    (10, 16.555124629)
    (100, 19.187612061)
    (1000, 36.835766243)
  };
  \addlegendentry{Walrus}
  \addplot[bluefte,dotted,mark=triangle*,mark size=2pt] coordinates {
    (1, 23.722778872)
    (10, 27.283519743)
    (100, 26.836098623)
    (1000, 53.263620179)
  };
  \addlegendentry{RS(3,2)}
  \addplot[bluefte,dashdotted,mark=diamond*,mark size=2pt] coordinates {
    (1, 24.63928701)
    (10, 24.13793654)
    (100, 25.017513062)
    (1000, 62.650280182)
  };
  \addlegendentry{RS(5,3)}
  \addplot[bluefte,densely dashed,mark=pentagon*,mark size=2pt] coordinates {
    (1, 26.133445184)
    (10, 22.521439077)
    (100, 23.130894983)
    (1000, 58.483587926)
  };
  \addlegendentry{RS(10,6)}
  \addplot[bluefte,loosely dotted,mark=o,mark size=2pt] coordinates {
    (1, 30.565757187)
    (10, 25.860078273)
    (100, 21.537907253)
    (1000, 45.106225806)
  };
  \addlegendentry{RS(10,9)}
  \end{axis}
\end{tikzpicture}
  \caption{End-to-end store latency. Lines show means of ten successful attempts per size and
    configuration; faint dots show individual observations. The data-size axis is logarithmic.}
  \label{fig:macro-store-latency}
\end{figure}
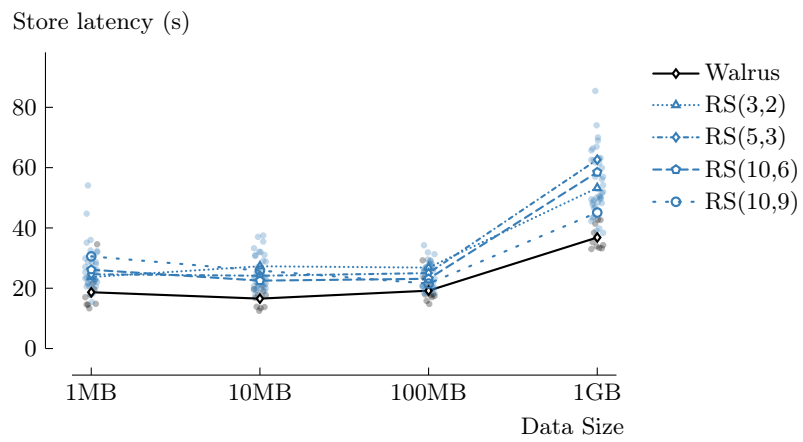

\FloatBarrier

\clearpage
\section{Robustness Metadata Overhead}\label{app:robustness-overhead}

We count authentication bytes per share, excluding payload, base key shares,
indices, and framing. Matching $f=\min(k-1,n-k)$ corruptions requires RS dimension
$d\leq n-2f$; we use $d=n-2f$. Unlike the baselines, we retune the code and
optimize the published layout~\cite{bellare2007robust,chen2017revisiting} by
regenerating the local fingerprint fragment instead of storing it. A faulty owner
still contributes only one erroneous position. With $32$-byte hashes/commitments
and $16$-byte openings, each remaining fragment takes $\lceil32/d\rceil$ bytes, giving
\begin{equation}\label{eq:metadata-padded}
  \begin{aligned}
    M_{\mathrm{HK1}}&=(n-1)\left\lceil\frac{32}{d}\right\rceil,
    &M_{\sysname}&=64+32=96,\\
    M_{\mathrm{HK2\text{-}style}}&=M_{\mathrm{RAONT\text{-}RS}}=M_{\mathrm{HK1}}+16.
  \end{aligned}
\end{equation}
These sizes describe our hash-based commitments, not formal HK2's statistically
hiding primitive; \sysname stores an Ed25519 signature and verification key.

\begin{figure}[ht]
  \centering
  \begingroup
\begin{tikzpicture}
\begin{groupplot}[
  group style={group size=2 by 1,horizontal sep=0.12\linewidth},
  scale only axis,width=0.38\linewidth,height=2.7cm,
  xmin=4,xmax=128,ymin=0,
  xlabel={Shares $n$},ylabel={Bytes/share},
  title style={font=\footnotesize,yshift=10pt},
  tick label style={font=\footnotesize},
  x label style={font=\footnotesize},y label style={font=\footnotesize},
  scaled y ticks=false,clip=false,
]
\nextgroupplot[title={(a) $k=n-1$},ymax=155,xtick={4,32,64,96,128},ytick={0,48,96,144},legend to name=robustness-metadata-legend,legend style={legend columns=3,font=\footnotesize,legend cell align=left,column sep=4pt,draw=none}]
\addplot+[redfte,mark=*,mark size=0.65pt] coordinates {(4,48) (6,40) (8,42) (10,36) (12,44) (14,39) (16,45) (18,34) (20,38) (22,42) (24,46) (26,50) (28,54) (30,58) (32,62) (34,33) (36,35) (38,37) (40,39) (42,41) (44,43) (46,45) (48,47) (50,49) (52,51) (54,53) (56,55) (58,57) (60,59) (62,61) (64,63) (66,65) (68,67) (70,69) (72,71) (74,73) (76,75) (78,77) (80,79) (82,81) (84,83) (86,85) (88,87) (90,89) (92,91) (94,93) (96,95) (98,97) (100,99) (102,101) (104,103) (106,105) (108,107) (110,109) (112,111) (114,113) (116,115) (118,117) (120,119) (122,121) (124,123) (126,125) (128,127)};
\addlegendentry{HK1}
\addplot+[yellowfte,mark=*,mark size=0.65pt] coordinates {(4,64) (6,56) (8,58) (10,52) (12,60) (14,55) (16,61) (18,50) (20,54) (22,58) (24,62) (26,66) (28,70) (30,74) (32,78) (34,49) (36,51) (38,53) (40,55) (42,57) (44,59) (46,61) (48,63) (50,65) (52,67) (54,69) (56,71) (58,73) (60,75) (62,77) (64,79) (66,81) (68,83) (70,85) (72,87) (74,89) (76,91) (78,93) (80,95) (82,97) (84,99) (86,101) (88,103) (90,105) (92,107) (94,109) (96,111) (98,113) (100,115) (102,117) (104,119) (106,121) (108,123) (110,125) (112,127) (114,129) (116,131) (118,133) (120,135) (122,137) (124,139) (126,141) (128,143)};
\addlegendentry{HK2-style / RAONT-RS}
\addplot+[black,mark=none,mark size=0.65pt] coordinates {(4,96) (6,96) (8,96) (10,96) (12,96) (14,96) (16,96) (18,96) (20,96) (22,96) (24,96) (26,96) (28,96) (30,96) (32,96) (34,96) (36,96) (38,96) (40,96) (42,96) (44,96) (46,96) (48,96) (50,96) (52,96) (54,96) (56,96) (58,96) (60,96) (62,96) (64,96) (66,96) (68,96) (70,96) (72,96) (74,96) (76,96) (78,96) (80,96) (82,96) (84,96) (86,96) (88,96) (90,96) (92,96) (94,96) (96,96) (98,96) (100,96) (102,96) (104,96) (106,96) (108,96) (110,96) (112,96) (114,96) (116,96) (118,96) (120,96) (122,96) (124,96) (126,96) (128,96)};
\addlegendentry{\sysname}
\addplot+[redfte,dashed,mark=none,forget plot] coordinates {(4,64) (6,48) (8,42.66666667) (10,40) (12,38.4) (14,37.33333333) (16,36.57142857) (18,36) (20,35.55555556) (22,35.2) (24,34.90909091) (26,34.66666667) (28,34.46153846) (30,34.28571429) (32,34.13333333) (34,34) (36,33.88235294) (38,33.77777778) (40,33.68421053) (42,33.6) (44,33.52380952) (46,33.45454545) (48,33.39130435) (50,33.33333333) (52,33.28) (54,33.23076923) (56,33.18518519) (58,33.14285714) (60,33.10344828) (62,33.06666667) (64,33.03225806) (66,33) (68,32.96969697) (70,32.94117647) (72,32.91428571) (74,32.88888889) (76,32.86486486) (78,32.84210526) (80,32.82051282) (82,32.8) (84,32.7804878) (86,32.76190476) (88,32.74418605) (90,32.72727273) (92,32.71111111) (94,32.69565217) (96,32.68085106) (98,32.66666667) (100,32.65306122) (102,32.64) (104,32.62745098) (106,32.61538462) (108,32.60377358) (110,32.59259259) (112,32.58181818) (114,32.57142857) (116,32.56140351) (118,32.55172414) (120,32.54237288) (122,32.53333333) (124,32.52459016) (126,32.51612903) (128,32.50793651)};
\addplot+[yellowfte,dashed,mark=none,forget plot] coordinates {(4,80) (6,64) (8,58.66666667) (10,56) (12,54.4) (14,53.33333333) (16,52.57142857) (18,52) (20,51.55555556) (22,51.2) (24,50.90909091) (26,50.66666667) (28,50.46153846) (30,50.28571429) (32,50.13333333) (34,50) (36,49.88235294) (38,49.77777778) (40,49.68421053) (42,49.6) (44,49.52380952) (46,49.45454545) (48,49.39130435) (50,49.33333333) (52,49.28) (54,49.23076923) (56,49.18518519) (58,49.14285714) (60,49.10344828) (62,49.06666667) (64,49.03225806) (66,49) (68,48.96969697) (70,48.94117647) (72,48.91428571) (74,48.88888889) (76,48.86486486) (78,48.84210526) (80,48.82051282) (82,48.8) (84,48.7804878) (86,48.76190476) (88,48.74418605) (90,48.72727273) (92,48.71111111) (94,48.69565217) (96,48.68085106) (98,48.66666667) (100,48.65306122) (102,48.64) (104,48.62745098) (106,48.61538462) (108,48.60377358) (110,48.59259259) (112,48.58181818) (114,48.57142857) (116,48.56140351) (118,48.55172414) (120,48.54237288) (122,48.53333333) (124,48.52459016) (126,48.51612903) (128,48.50793651)};
\addlegendimage{greyfte,solid,mark=*}
\addlegendentry{Optimized per-hash}
\addlegendimage{greyfte,dashed,mark=none}
\addlegendentry{Packed (no omission)}
\draw[yellowfte,densely dotted,thin] (axis cs:82,0) -- (axis cs:82,111);
\node[font=\footnotesize,anchor=south east,yellowfte] at (axis cs:82,111) {$n=82$};
\addplot+[yellowfte,only marks,mark=*,mark size=1.5pt,forget plot] coordinates {(82,97)};
\draw[redfte,densely dotted,thin] (axis cs:98,0) -- (axis cs:98,137);
\node[font=\footnotesize,anchor=south east,redfte] at (axis cs:98,137) {$n=98$};
\addplot+[redfte,only marks,mark=*,mark size=1.5pt,forget plot] coordinates {(98,97)};
\nextgroupplot[title={(b) $k=n/2+1$},ymax=2200,xtick={4,32,64,96,128},ytick={0,1000,2000}]
\addplot+[redfte,mark=*,mark size=0.65pt] coordinates {(4,48) (6,80) (8,112) (10,144) (12,176) (14,208) (16,240) (18,272) (20,304) (22,336) (24,368) (26,400) (28,432) (30,464) (32,496) (34,528) (36,560) (38,592) (40,624) (42,656) (44,688) (46,720) (48,752) (50,784) (52,816) (54,848) (56,880) (58,912) (60,944) (62,976) (64,1008) (66,1040) (68,1072) (70,1104) (72,1136) (74,1168) (76,1200) (78,1232) (80,1264) (82,1296) (84,1328) (86,1360) (88,1392) (90,1424) (92,1456) (94,1488) (96,1520) (98,1552) (100,1584) (102,1616) (104,1648) (106,1680) (108,1712) (110,1744) (112,1776) (114,1808) (116,1840) (118,1872) (120,1904) (122,1936) (124,1968) (126,2000) (128,2032)};
\addplot+[yellowfte,mark=*,mark size=0.65pt] coordinates {(4,64) (6,96) (8,128) (10,160) (12,192) (14,224) (16,256) (18,288) (20,320) (22,352) (24,384) (26,416) (28,448) (30,480) (32,512) (34,544) (36,576) (38,608) (40,640) (42,672) (44,704) (46,736) (48,768) (50,800) (52,832) (54,864) (56,896) (58,928) (60,960) (62,992) (64,1024) (66,1056) (68,1088) (70,1120) (72,1152) (74,1184) (76,1216) (78,1248) (80,1280) (82,1312) (84,1344) (86,1376) (88,1408) (90,1440) (92,1472) (94,1504) (96,1536) (98,1568) (100,1600) (102,1632) (104,1664) (106,1696) (108,1728) (110,1760) (112,1792) (114,1824) (116,1856) (118,1888) (120,1920) (122,1952) (124,1984) (126,2016) (128,2048)};
\addplot+[black,mark=none,mark size=0.65pt] coordinates {(4,96) (6,96) (8,96) (10,96) (12,96) (14,96) (16,96) (18,96) (20,96) (22,96) (24,96) (26,96) (28,96) (30,96) (32,96) (34,96) (36,96) (38,96) (40,96) (42,96) (44,96) (46,96) (48,96) (50,96) (52,96) (54,96) (56,96) (58,96) (60,96) (62,96) (64,96) (66,96) (68,96) (70,96) (72,96) (74,96) (76,96) (78,96) (80,96) (82,96) (84,96) (86,96) (88,96) (90,96) (92,96) (94,96) (96,96) (98,96) (100,96) (102,96) (104,96) (106,96) (108,96) (110,96) (112,96) (114,96) (116,96) (118,96) (120,96) (122,96) (124,96) (126,96) (128,96)};
\addplot+[redfte,dashed,mark=none,forget plot] coordinates {(4,64) (6,96) (8,128) (10,160) (12,192) (14,224) (16,256) (18,288) (20,320) (22,352) (24,384) (26,416) (28,448) (30,480) (32,512) (34,544) (36,576) (38,608) (40,640) (42,672) (44,704) (46,736) (48,768) (50,800) (52,832) (54,864) (56,896) (58,928) (60,960) (62,992) (64,1024) (66,1056) (68,1088) (70,1120) (72,1152) (74,1184) (76,1216) (78,1248) (80,1280) (82,1312) (84,1344) (86,1376) (88,1408) (90,1440) (92,1472) (94,1504) (96,1536) (98,1568) (100,1600) (102,1632) (104,1664) (106,1696) (108,1728) (110,1760) (112,1792) (114,1824) (116,1856) (118,1888) (120,1920) (122,1952) (124,1984) (126,2016) (128,2048)};
\addplot+[yellowfte,dashed,mark=none,forget plot] coordinates {(4,80) (6,112) (8,144) (10,176) (12,208) (14,240) (16,272) (18,304) (20,336) (22,368) (24,400) (26,432) (28,464) (30,496) (32,528) (34,560) (36,592) (38,624) (40,656) (42,688) (44,720) (46,752) (48,784) (50,816) (52,848) (54,880) (56,912) (58,944) (60,976) (62,1008) (64,1040) (66,1072) (68,1104) (70,1136) (72,1168) (74,1200) (76,1232) (78,1264) (80,1296) (82,1328) (84,1360) (86,1392) (88,1424) (90,1456) (92,1488) (94,1520) (96,1552) (98,1584) (100,1616) (102,1648) (104,1680) (106,1712) (108,1744) (110,1776) (112,1808) (114,1840) (116,1872) (118,1904) (120,1936) (122,1968) (124,2000) (126,2032) (128,2064)};
\draw[black,densely dotted,thin] (axis cs:8,0) -- (axis cs:8,1750);
\node[font=\footnotesize,anchor=south west,black] at (axis cs:8,1750) {Both: $n=8$};
\addplot+[redfte,only marks,mark=*,mark size=1.5pt,forget plot] coordinates {(8,112)};
\addplot+[yellowfte,only marks,mark=*,mark size=1.5pt,forget plot] coordinates {(8,128)};
\end{groupplot}
\end{tikzpicture}
\par\smallskip
\pgfplotslegendfromname{robustness-metadata-legend}
\endgroup
  \caption{Matched-tolerance authentication costs. Solid: optimized per-hash;
    dashed: ideal packing without local omission. Labels mark first strict wins
    for \sysname at even $n$.}
  \label{fig:robustness-metadata}
\end{figure}
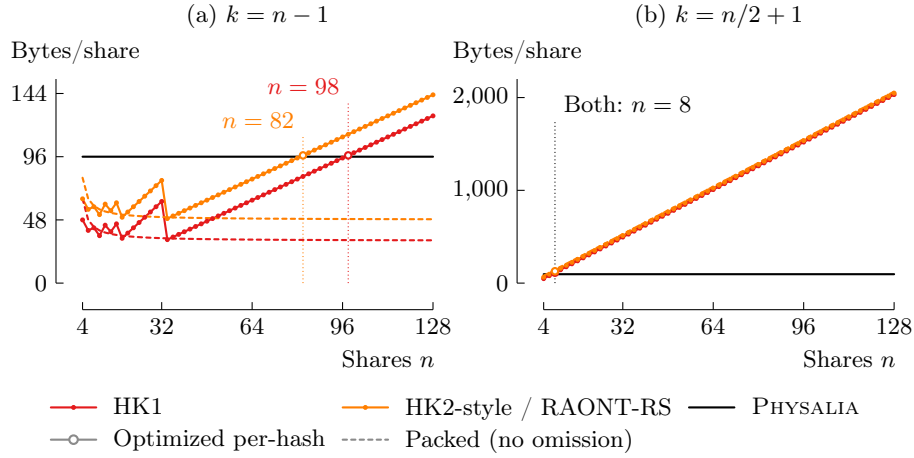

\emph{High threshold ($k=n-1$).}
Here $d=n-2$. For $n\geq34$, the optimized costs are $n-1$ and $n+15$ bytes,
giving first even-$n$ wins at $98$ and $82$. This is a padding artifact:
ideal packing of all $n$ values without local omission costs $32n/(n-2)$ bytes
(plus $16$ for openings), below $96$ for $n\geq4$.

\emph{Near-majority threshold ($k=n/2+1$, even $n$).}
Here $d=2$: HK1 needs $16(n-1)$ bytes and HK2-style/RAONT-RS $16n$.
Both grow linearly; \sysname is strictly smaller from even $n=8$, with a tie
against HK2-style/RAONT-RS at $n=6$. Packing also remains linear.
\FloatBarrier

\end{document}